\def\fullversionflag{1}     

\ifodd\fullversionflag
  \documentclass[letterpaper,11pt]{llncs}
  \usepackage{fullpage}
\else
  \documentclass[envcountsect,envcountsame]{llncs}
\fi

\usepackage{etex}
\usepackage{etoolbox}

\newtoggle{fullversion}
\ifodd\fullversionflag
    \toggletrue{fullversion}
\else
    \togglefalse{fullversion}
\fi

\newcommand{\ignore}[1]{}
\usepackage{float}
\usepackage{enumitem}
\usepackage{amsmath,amsfonts,amssymb,amstext,mathtools}
\usepackage{latexsym,ifthen,url,bm,stmaryrd}
\usepackage[colorlinks=true,citecolor=purple]{hyperref}
\usepackage[usenames,svgnames,table]{xcolor}
\usepackage[probability,operators,sets,lambda,advantage,adversary,keys,asymptotics,mm,logic]{cryptocode}
\usepackage[many,breakable]{tcolorbox}
\usepackage{mdframed}
\usepackage{bbm}
\usepackage{algorithm}
\usepackage{algpseudocode}
\usepackage{graphicx}
\usepackage{subcaption}
\usepackage{placeins}

\newcommand{\deq}{\mathrel{\mathop:}=}

\newcommand{\paragraphb}[1]{\smallskip\noindent{\bf #1.}}

\newcommand{\prompt}{\mathsf{prompt}}
\newcommand{\terminate}{\mathsf{ter}}
\newcommand{\PRFEval}{\mathsf{PRF}.\mathsf{Eval}}
\newcommand{\PRFGen}{\mathsf{PRF}.\mathsf{Gen}}
\newcommand{\PRF}{\mathsf{PRF}}
\newcommand{\KG}{\mathsf{KeyGen}}
\newcommand{\Enc}{\mathsf{Enc}}
\newcommand{\Dec}{\mathsf{Dec}}
\newcommand{\Perturb}{\mathsf{Perturb}}

\title{Robust Steganography from Large Language Models}

\iftoggle{fullversion}{
    \usepackage{authblk}
    \author{Neil Perry}
    \author{Sanket Gupte}
    \author{Nishant Pitta}
    \author{Lior Rotem}
    
    \affil{Stanford University\\ {\tt\small \{naperry,sanketg,lrotem\}@cs.stanford.edu}}
    \institute{}
    \date{}
}{
\author{}
\institute{}
}

\begin{document}

\maketitle
\iftoggle{fullversion}{}{\vspace{-1cm}}

\begin{abstract}

Recent steganographic schemes, starting with Meteor (CCS'21), rely on leveraging large language models (LLMs) to resolve a historically-challenging task of disguising covert communication as ``innocent-looking'' natural-language communication. 
However, existing methods are vulnerable to ``re-randomization attacks,'' where slight changes to the communicated text, that might go unnoticed, completely destroy any hidden message.  
This is also a vulnerability in more traditional encryption-based stegosystems, 
where adversaries can modify the randomness of an encryption scheme to destroy the hidden message while preserving an acceptable covertext to ordinary users. 
In this work, we study the problem of {\bf robust steganography}. 
We introduce formal definitions of weak and strong robust LLM-based steganography, corresponding to two threat models in which natural language serves as a covertext channel resistant to realistic re-randomization attacks. We then propose two constructions satisfying these notions. We design and implement our steganographic schemes that embed arbitrary secret messages into natural language text generated by LLMs, ensuring recoverability even under adversarial paraphrasing and rewording attacks. To support further research and real-world deployment, we release our implementation and datasets for public use.
\end{abstract}

\section{Introduction}
\label{introduction}

In authoritarian regimes, even the mere act of encryption can be a death sentence. In such settings, \emph{steganography}---the art of embedding hidden messages within innocent-looking communications~\cite{petitcolas2002information,johnson1998exploring}---provides a critical lifeline. While cryptography secures content against unauthorized readers, steganography goes further by concealing the very existence of the communication. In repressive environments, this distinction is critical: even the act of encrypting a message can raise suspicion or trigger punitive measures.

To ground this discussion, imagine a world where a dictatorship mandates that all communication be conducted using state-issued public-key encryption, and the regime retains a copy of every citizen's secret key. In such a setting, encryption offers no privacy—messages are effectively plaintext to the state. Yet remarkably, Horel~et~al.~\cite{steganography_horel} showed that even under such extreme surveillance, hidden communication is still possible. By leveraging the randomness inherent in encryption, one can encode secret bits via rejection sampling, choosing ciphertexts whose structure encodes information while remaining indistinguishable from ordinary ciphertexts.

However, such techniques are fragile. A simple countermeasure available to the regime is ciphertext re-randomization: decrypt the message and re-encrypt the same plaintext with fresh randomness before forwarding it. This effectively erases the hidden message, without altering the visible content of the communication. Similar fragility appears in modern steganographic schemes based on large language models (LLMs), where information is hidden in sampled tokens conditioned on a secret seed \cite{meteor}. These methods can also be disrupted by semantically neutral edits—such as rewording, paraphrasing, or reformatting—that preserve the apparent meaning of the message but destroy the encoding.

This brittleness motivates a new desideratum for steganographic systems: \emph{robustness}. Intuitively, a robust stegosystem should withstand realistic modifications to the covertext, preserving the hidden message so long as the overt meaning remains unchanged. This models the behavior of a sophisticated adversary—e.g., a censorship regime—that seeks to strip hidden content while maintaining the flow of normal communication.

In this paper, we initiate the study of robust steganography. We formalize adversarial models in which covertexts are subject to modification—ranging from localized perturbations to full semantic paraphrasing—and define both weak and strong notions of robustness corresponding to these capabilities.

We then present two new steganographic constructions that aim to meet these definitions. The first perturbs the sampling distribution of an LLM to encode hidden bits in a key-dependent but undetectable manner, and is weakly robust to local edits. The second uses semantic embeddings and locality-sensitive hashing to embed information in the semantic space of the message, enabling strong robustness to paraphrasing and meaning-preserving edits.

Both constructions are implemented and evaluated. We measure their resistance to realistic tampering, assess the overhead required to achieve robustness, and provide practical guidance for deployment. Our watermarking-based system survives most classes of attacks—even when up to 20\% of the covertext is modified—while maintaining a recovery rate of approximately 90\% or higher. The embedding-based scheme resists attacks far more effectively than the watermarking-based schemes-- surviving up to 50\% covertext modifications while requiring 100x fewer tokens for the covertext per bit hidden. Additionally, the embedding system on average costs $<$ 1¢ per hidden byte of data. Our code and experiments are available at \url{https://github.com/NeilAPerry/robust_steganography}.

\section{Background and Related Work}

Steganography—the practice of hiding secret information within other innocent-looking messages—has long been studied as a means of secure communication in adversarial environments. In this paper, we consider a particularly extreme case: a world in which a dictatorship assigns each citizen a public-secret key pair, mandates the use of these keys for all communication, and retains a copy of every citizen's secret key. In such a setting, the dictator can read all communications as if they were transmitted in plaintext.

Horel~et~al.~\cite{steganography_horel} recently demonstrated that even under such draconian conditions, private, undetectable communication remains possible. Their construction operates by selecting the encryption randomness in a way that depends on the secret message. More concretely, suppose Alice wishes to send a secret bit $b \in \{0,1\}$ to Bob, with whom she shares a secret key $k$ unknown to the dictator. She repeatedly encrypts an innocent-looking message $m$ with fresh randomness until she obtains a ciphertext $ct$ such that $H(k,ct) = b$, where $H$ is a suitable hash function. This method is effective provided it meets two key criteria:

\begin{itemize}
	\item The ciphertext should remain indistinguishable from a normal encryption of $m$ to the dictatorship, even though it possesses Bob's secret key (but not $k$); and
	\item For any given $m$ and $b$, a ciphertext $ct$ satisfying $H(k,ct) = b$ should be found after only a small number of trials.
\end{itemize}

Horel et al. showed that both conditions can be met under standard cryptographic assumptions—specifically, that the encryption scheme is semantically secure to any party lacking the recipient’s secret key. They further demonstrated that similar ideas can be used to undetectably execute key exchange protocols, even in this hostile setting.

The general idea of using rejection sampling over a high-entropy stochastic channel for steganography dates back to early foundational work~\cite{cachin1998information}, and variations of this idea have appeared in nearly all subsequent rigorous stegosystems~\cite{hopper2002provably,von2004public,backes2005public}.

However, such schemes are inherently fragile. A straightforward countermeasure available to the dictatorship is ciphertext re-randomization: the regime decrypts the ciphertext $ct$, re-encrypts the same plaintext $m$ with fresh randomness, and forwards the new ciphertext $\hat{ct}$ to Bob. This re-randomization destroys any steganographically embedded content, while preserving the plaintext, and therefore does not interfere with honest communication. In light of this, we say that the scheme of Horel et al.\ lacks \emph{robustness}, since ``reasonable modifications’’ to the covertext—in this case, the ciphertext—destroy the hidden message. We formally define robustness and reasonable modifications in Section~\ref{adversary_capabilities}.

The situation worsens if the regime bans randomized encryption altogether, mandating the use of deterministic encryption~\cite{bellare2007deterministic,rogaway2006provable,brakerski2011better}, or even prohibits encryption entirely. In such settings, any steganographic embedding must occur directly in the plaintext channel. This possibility has prompted extensive research into steganographic schemes that embed hidden messages within natural language text~\cite{wayner1992mimic,chapman1997hiding,bennett2004linguistic,topkara2007information,shirali2008text,yang2018rnn,ziegler2019neural,zhang2021provably}, or other natural data such as images~\cite{image_steg,westfeld2001f5,cheddad2010digital,baluja2017hiding,zhu2018hidden,lu2021large}.

The idea of hiding messages within common language is not new—humans have been doing so for thousands of years, with historical accounts dating back to Ancient Greece~\cite{stego_history}. However, nearly all such systems, including the modern ones cited above, are heuristic in nature and lack rigorous security guarantees. Their fragility arises from two central challenges recognized in the formal study of steganography:

\begin{itemize}
	\item A well-defined probability distribution over the vocabulary of the covertext language is required, conditioned on prior history; and
	\item Each block of text must contain sufficient entropy to allow secure hiding of information.
\end{itemize}

Natural languages such as English pose problems for both requirements. The distribution over valid continuations of a sentence is hard to formalize, highly context-dependent, and typically exhibits low entropy in predictable settings. For instance, the sentence "At the museum I saw a Tyrannosaurus" is overwhelmingly likely to be completed with "Rex," offering little room for information embedding due to the determinism of the next token.

Meteor~\cite{meteor} addresses these challenges by using modern large language models (LLMs) to both produce and analyze text. In Meteor, both sender and receiver rely on the same LLM to model token distributions and determine where information can be safely embedded. Specifically, the sender uses a ciphertext as the seed to condition the LLM and begins generating text from that state. At each token step, the model's predicted distribution is used to assess whether sufficient entropy is present to embed bits. If so, the sender samples from a subset of the distribution in a way that encodes the secret message. The receiver, using the same LLM, can detect high-entropy positions and recover the embedded bits.

Meteor’s approach offers several advantages. It ensures that steganographic messages are only embedded when statistically safe, that the covertext closely resembles human writing, and that both sender and receiver can independently determine where embedding is feasible. Nonetheless, Meteor inherits the fundamental vulnerability of all previous systems discussed thus far: it is not robust. Minor changes to the generated text—such as replacing words with synonyms, adding or removing words, or rephrasing sentences—are enough to alter the decoding trajectory of the LLM and make the original seed unrecoverable. These edits are commonplace in communication and do not affect semantic meaning, yet they completely destroy the hidden information. Thus, like the earlier rejection sampling schemes, Meteor fails to support robust steganography.

Similar vulnerabilities appear in other systems that embed covert data into low-level features of ciphertexts or generated text. Wink~\cite{wink} is a system that embeds hidden messages within the random values of encrypted communications, achieving deniability by ensuring that ciphertexts can plausibly be interpreted in multiple ways. However, like Horel et al.'s construction, Wink is fragile in the face of re-randomization attacks: an adversary can simply decrypt and re-encrypt a message, thereby removing any embedded content without altering the plaintext. Since these re-randomized ciphertexts remain valid and indistinguishable from ordinary ones, this attack is undetectable and does not interfere with honest communication.

Other approaches in linguistic steganography, such as SegFree~\cite{segfree}, attempt to embed covert messages directly within natural language output by using generative methods. 
Additionally, recent work such as Discop \cite{discop} has explored provably secure steganographic methods that ensure covertext distributions remain indistinguishable from specified original distributions, highlighting the importance of statistical matching in steganography.

The broader literature on watermarking LLM outputs provides further evidence of how fragile these systems can be. Recently established impossibility results for strong watermarking in generative models show that certain watermarking schemes can be removed without significant degradation to the generated content \cite{watermarks_in_sand}. However, these results do not apply to our system, as our watermarking-based approach has stricter requirements than the general case.

Christ~and~Gunn~\cite{christ2024pseudorandom}, done concurrently to this work, use watermarks to embed hidden bits in text via pseudorandom error-correcting codes, a technique similar to our multi-watermark embedding method. Zamir~\cite{zamir2024excuse}, done concurrently to this work, further extends this and presents a steganographic method embedding hidden messages in LLM outputs while ensuring undetectability, which is similar to our multi-watermark embedding approach. Crucially, it is not robust to paraphrasing text. Our embeddings based approach is robust to paraphrasing attacks.

Taken together, the prior work discussed here paints a clear picture: the state of the art in steganography—across ciphertexts, language, and watermarking—remains vulnerable to simple, realistic tampering. Existing systems emphasize undetectability and decoding under ideal conditions but offer little protection against covertext modifications. Robustness—the ability to recover hidden messages after adversarial edits that preserve semantics—has received scant attention and is largely absent from formal definitions and constructions. In this work, we initiate a systematic study of robustness in steganographic systems and demonstrate that meaningful progress is possible in both theory and practice.

\FloatBarrier
\section{Preliminaries}

\subsection{Pseudorandom Functions}
We use the following standard notion of a pseudorandom function. Let $\PRF = (\PRFGen, \PRFEval)$ be a function family over domain $\{ \mathcal{X}_\lambda\}_{\lambda \in \mathbb{N}}$ with range $\{ \mathcal{Y}_\lambda\}_{\lambda \in \mathbb{N}}$ and key space $\{ \mathcal{K}_\lambda\}_{\lambda \in \mathbb{N}}$, such that:
\begin{itemize}
	\item $\PRFGen$ is a probabilistic polynomial-time algorithm, which takes as input the security parameter $\lambda \in \mathbb{N}$ and outputs a key $k \in \mathcal{K}_\lambda$.
	\item $\PRFEval$ is a deterministic polynomial-time algorithm, which takes as input a key $k\in\mathcal{K}_\lambda$ and a domain element $x \in \mathcal{X}_\lambda$ and outputs a value $y \in \mathcal{Y}_\lambda$.
\end{itemize}%
For ease of notation, for a key $k\in \mathcal{K}_\lambda$, we denote by $\PRF_K(\cdot)$ the function $\PRFEval(K,\cdot)$. We also assume without loss of generality that for every $\lambda \in \mathbb{N}$, it holds that $\mathcal{K}_\lambda = \{0,1\}^\lambda$ and that $\PRFGen(1^\lambda)$ simply samples $k$ from $\{0,1\}^\lambda$ uniformly at random. Using these conventions, the following definition captures the standard notion of a pseudorandom function family.

\begin{definition}\label{Def:PRF}
	A function family $\PRF = (\PRFGen, \PRFEval)$ is {\em pseudorandom} if for every probabilistic polynomial-time algorithm $D$, there exists a negligible function $\nu(\cdot)$ such that  
	\[ \advantage{prf}{D, \PRF} \deq \left|   \Pr_{k \gets \{0,1\}^\lambda} \left[  {\sf D}(1^\lambda)^{\PRF_k(\cdot)} = 1  \right]  -  	\Pr_{f \gets \mathcal{F}_\lambda} \left[  {\sf D}(1^\lambda)^{f(\cdot)} = 1  \right] \right| \leq \nu(\lambda),  \]%
	for all sufficiently large $\lambda \in \mathbb{N}$, where $\mathcal{F}_\lambda$ is the set of all functions mapping $\mathcal{X}_\lambda$ into $\mathcal{Y}_\lambda$.
\end{definition}

\subsection{Language Models}

A {\bf language model} is a pair $(\mathcal{T},M)$, where $\mathcal{T} = \{ t_1,\ldots, t_n \}$ is a finite set of symbols called {\em tokens}, and $M$ is a probabilistic polynomial-time algorithm. The set $\mathcal{T}$ includes the special termination symbol $\terminate$. $M$ takes as input a {\em prompt} $\prompt \in \mathcal{T}^{\ast}$ and a (possibly empty) sequence $\Vec{w} = (w_1,\ldots, w_{i-1}) \in \mathcal{T}^{i-1}$ of tokens.\footnote{For ease of presentation, we assume that $\prompt \in \mathcal{T}^{\ast}$. However, everything in this paper readily extends to the case where the prompt may contain symbols not in $\mathcal{T}$.} On this input, $M$ outputs a probability distribution $P_i \gets M(\prompt, \Vec{w})$ over $\mathcal{T}$. That is, for every $t_j \in \mathcal{T}$, $P_i(t_j) \in [0,1]$ and $\sum_{j \in [n]} P_i(t_j) = 1$.

For a language model $(\mathcal{T},M)$ and a prompt $\prompt \in \mathcal{T}^{\ast}$, we may consider the following probability distribution over $\mathcal{T}^\ast$: 
\begin{enumerate}
    \item Compute $P_1 \gets M(\prompt)$, and sample $w_1 \sample P_1$. If $w_1 = \terminate$, then output $w_1$ and terminate. 
    \item Set $i \gets 1$ and $\Vec{w} \gets (w_1)$. 
    \item Until $w_i = \terminate$ do:
    \begin{enumerate}
        \item Increment $i \gets i+1$.
        \item Compute $P_i \gets M(\prompt, \Vec{w})$, and sample $w_i \sample P_i$. 
        \item Append $w_i$ to $\Vec{w}$.
    \end{enumerate}
    \item Output $\Vec{w}$.
\end{enumerate}
We denote the process of sampling from this distribution by $\Vec{w} \sample M^\ast(\prompt)$, and call $\Vec{w}$ a {\em response} of $M$ on $\prompt$.

\subsection{Locally Sensitive Hashing}

{\bf Locally Sensitive Hashing (LSH)} is a technique used to map high-dimensional data into a lower-dimensional space while preserving proximity relationships between data points. Formally, a family of hash functions $\mathcal{H} = \{h : \mathbb{R}^d \to \{0,1\}^n\}$ is called $(r_1, r_2, p_1, p_2)$-sensitive for a similarity function $S : \mathbb{R}^d \times \mathbb{R}^d \to [0, 1]$ if, for any $x, y \in \mathbb{R}^d$:
\begin{itemize}
    \item If $S(x, y) \geq r_1$, then $\Pr[h(x) = h(y)] \geq p_1$,
    \item If $S(x, y) \leq r_2$, then $\Pr[h(x) = h(y)] \leq p_2$,
\end{itemize}
where $r_1 > r_2$ and $p_1 > p_2$.

The output of an LSH function is an $n$-bit binary string, where each bit is determined based on the properties of the hash function and the similarity of the input data points. LSH is commonly used in applications requiring approximate nearest neighbor search, as it efficiently maps similar inputs to similar hash outputs with high probability.

\subsection{Embeddings}

{\bf Embeddings} are dense vector representations of objects, such as words, sentences, or paragraphs, in a high-dimensional space. An embedding function $E : \mathcal{T}^\ast \to \mathbb{R}^d$ maps a sequence of tokens $t_1, t_2, \ldots, t_n$ to a vector in $\mathbb{R}^d$, where $d$ is the dimensionality of the embedding space.

Embeddings encode semantic and contextual information about the objects they represent. For example, embeddings generated by language models often reflect both lexical and syntactic relationships, enabling comparisons between similar or related objects. Formally, given two sequences $x, y \in \mathcal{T}^\ast$, their similarity can often be evaluated using geometric measures such as cosine similarity or Euclidean distance between their embeddings.

The embedding space is characterized by its high dimensionality and structure, which enables it to represent a wide range of linguistic and semantic features compactly and effectively.
\FloatBarrier
\section{Threat Model}

Robust steganographic schemes must withstand adversaries capable of intercepting and modifying communications. In this section, we define the adversary’s goals, capabilities, and constraints. To reflect different levels of adversarial sophistication, we introduce two threat models—a {\bf weak model} and a {\bf strong model}—that motivate our subsequent definitions of robustness.

\subsection{Adversary Motivation and Justification}
The adversary seeks to maintain control over information flow by targeting hidden communication. Their primary objectives are:
\begin{itemize}
    \item {\bf Detection}: Identifying whether a message contains hidden information, enabling punitive or censorial actions.
    \item {\bf Disruption}: Modifying the message in a way that prevents successful decoding of hidden information by the intended recipient while leaving the meaning of the message intact.
\end{itemize}

However, the adversary must balance their actions against the potential for societal backlash. Excessive tampering that fundamentally disrupts the intent or meaning of messages can lead to public discontent or retaliation. Complete censorship of communication has been found to increase protests \cite{rydzak2020internet} and be prohibitively expensive \cite{robinson2013collateral}. Governments have been found to deploy selective censorship to achieve particular goals such as stability instead of censoring anything critical of the government \cite{king2013censorship}. For example:
\begin{itemize}
    \item If citizens notice that their everyday communications are being significantly altered or misrepresented, trust in the communication system may collapse.
    \item Such disruptions could incite rebellion or other forms of resistance, undermining the adversary’s goals.
\end{itemize}

Thus, the adversary seeks to strike a balance: they aim to detect or destroy hidden messages without making alterations so severe that they provoke widespread dissent. This motivates the formalization of tampering functions and defines the scope of robustness required for steganographic systems.

\subsection{Adversary Capabilities}
\label{adversary_capabilities}
We assume a man-in-the-middle adversary with the following capabilities:
\begin{itemize}
    \item {\bf Intercept and Analyze}: The adversary can intercept and analyze all communication between the sender and receiver.
    \item {\bf Modify Stegotext}: The adversary can modify the stegotext, either by introducing small structural changes or by generating semantically equivalent paraphrases.
    \item {\bf Linguistic Coherence}: The adversary has the ability to produce plausible modifications using automated or manual natural language processing techniques.
\end{itemize}

Despite these capabilities, the adversary is constrained by the need to avoid provoking widespread dissent among the affected population. This leads to two adversarial models: the {\bf weak threat model}, which assumes limited tampering capabilities, and the {\bf strong threat model}, which assumes more sophisticated attacks.

\subsection{Weak Threat Model: Locally Consistent Tampering}
\label{weak_threat_model}
This model corresponds to \emph{weak robustness} or \emph{local robustness}, which captures adversaries that make small, localized edits—such as reordering n-grams, substituting synonyms, or inserting short phrases. These changes typically preserve surface fluency and often preserve meaning, but this is not always the case. For example, inserting the word ``not'' may leave most of the structure unchanged while inverting the semantic intent of the sentence.

In the weak threat model, the adversary’s modifications are limited to small structural changes to the message. These changes are modeled by a family of tampering functions $\mathcal{F}$, specifically $(k, \epsilon)$-locally consistent functions.

\paragraphb{Definition of Locally Consistent Tampering} A tampering function $f \in \mathcal{F}_{k, \epsilon}$ ensures that for every input string $x \in \Sigma^\ast$, the tampered string $f(x)$ agrees with $x$ on at least an $\epsilon$-fraction of its $k$-length continuous substrings. Formally:
\[
\mathcal{F}_{k,\epsilon} = \left\{ f : \forall x\in \Sigma^\ast, \sum_{i = 1}^{|x|-k+1} \xi\left(x[i:i+k-1], f(x)\right) \geq \epsilon \cdot (|x| - k + 1) \right\},
\]
where $\xi(x[i:i+k-1], f(x)) = 1$ if $x[i:i+k-1]$ is a continuous substring of $f(x)$, and 0 otherwise.

Examples of locally consistent tampering include:
\begin{itemize}
    \item Reordering or shuffling $n$-grams.
    \item Deleting or repeating small parts of the message.
    \item Making minor edits that leave the message’s intent intact.
\end{itemize}

The weak model reflects adversarial actions designed to destroy hidden communication while avoiding overt disruption that might provoke dissatisfaction.

\subsection{Strong Threat Model: Semantically Preserving Paraphrasing}
\label{strong_threat_model}
This model corresponds to \emph{strong robustness} or \emph{semantic robustness}, which captures adversaries that make broader, meaning-preserving rewrites—such as paraphrasing, stylistic reformulations, or structural transformations. These changes may significantly reshape how the sentence is written while leaving its meaning intact.

The strong threat model assumes a more sophisticated adversary capable of generating paraphrases of the stegotext. These paraphrases:
\begin{itemize}
    \item {\bf Preserve Semantic Meaning}: The adversary rewrites the text to retain the original meaning, ensuring it remains plausible to the recipient.
    \item {\bf Maintain Embedding Proximity}: The paraphrase must remain close to the original text in the embedding space, reflecting semantic equivalence.
\end{itemize}

\paragraphb{Definition of Embedding Proximity} Let $E(x)$ denote the embedding of a text $x \in \Sigma^\ast$.

We define a distance function
\[
\mathrm{distance}: \mathbb{R}^d \times \mathbb{R}^d \to \mathbb{R}_{\ge 0},
\]
mapping pairs of embeddings in $\mathbb{R}^d$ to nonnegative real values. Any well-defined distance function (e.g., Euclidean distance or cosine similarity) may be used, and will correspond to a particular Locality Sensitive Hash (LSH) function.

A paraphrase $x'$ satisfies:
\[
\text{distance}(E(x), E(x')) \leq \delta,
\]
where $\delta > 0$ is a small threshold representing allowable semantic deviation.

Examples of adversarial actions in this model include:
\begin{itemize}
    \item Rewriting sentences using synonyms or alternative phrasing.
    \item Using NLP tools to generate semantically equivalent variations of the text.
    \item Making changes to sentence structure while retaining overall meaning.
\end{itemize}

This strong model reflects adversaries capable of leveraging advanced linguistic and computational techniques to detect or disrupt hidden messages.

\subsection{Summary and Comparison of Threat Models}
\begin{itemize}
    \item {\bf Weak Model (Local Robustness)}: The adversary modifies stegotext using $(k, \epsilon)$-locally consistent tampering functions, introducing small structural changes to the text.
    \item {\bf Strong Model (Semantic Robustness)}: The adversary generates semantically equivalent paraphrases constrained by proximity in the embedding space, representing a more sophisticated attack.
\end{itemize}

These threat models highlight the balance adversaries seek to strike: they aim to disrupt hidden communication without sufficiently disrupting the message’s meaning to provoke retaliation. The weak model aligns with scenarios involving simple structural tampering, while the strong model reflects adversaries equipped with advanced paraphrasing capabilities.
Although these definitions are formally incomparable—neither strictly implies the other—they differ in how well they align with the motivations of real-world adversaries. For instance, inserting the word ``not'' may satisfy the weak definition by preserving local structure, but it changes the meaning of the message and would likely be avoided by an adversary under the motivating conditions for this work. In this sense, the strong (semantic) definition more accurately captures the kinds of transformations we aim to defend against.
We retain the terms \textit{weak} and \textit{strong} throughout the paper to match our formal definitions and maintain continuity with prior work. At the same time, we use \textit{local} and \textit{semantic} as alternative names to clarify the nature of the edits each model targets. Notably, the weak model connects naturally to existing paradigms like watermarking, which can be seen as a special case of steganography and has historically focused on robustness to small, localized changes. Our strong definition generalizes this by targeting meaning-preserving transformations, motivating a new class of robust stegosystems.

\FloatBarrier
\section{Defining Robust Steganography}

In this section, we formally define the notion of a robust steganographic scheme. While standard security definitions ensure that hidden communication is indistinguishable from ordinary covertext, they fail to address scenarios where adversaries actively modify stegotext messages. As described in the threat model, adversaries aim to detect or disrupt hidden communication without provoking societal backlash, constraining their tampering to realistic limits. 

To address these practical considerations, we introduce robustness definitions that guarantee successful decoding under two levels of adversarial tampering:
\begin{itemize}
    \item {\bf Weak robustness}, which accounts for small structural changes modeled by locally consistent tampering functions.
    \item {\bf Strong robustness}, which accounts for semantically equivalent paraphrasing constrained by proximity in the embedding space.
\end{itemize}
These definitions provide formal guarantees against adversarial modifications that align with the weak and strong threat models. We first recall the standard definition of a symmetric-key steganographic scheme, and then define robustness under both the weak and strong threat models introduced earlier.

\subsection{Symmetric-key Steganography}

In this work, we focus on the symmetric-key setting. In this setting, the sender and receiver share a secret key $k$, and communicate over a {\em covertext channel} $Ch$. This channel is characterized by an alphabet $\Sigma$ and a mapping $\mathcal{D}$ that maps a history $\vec{h} = (h_1, \ldots, h_\ell)$ to a probability distribution over $\Sigma^\ast$.

Following the definitions in \cite{hopper2002provably}, a symmetric-key steganographic scheme is parameterized by a {\em covertext channel} with alphabet $\mathcal{T}$ and distribution $\mathcal{D}$, and consists of a tuple $\Pi = (\KG, \Enc, \Dec)$, where:
\begin{itemize}
    \item $\KG(\secparam) \to k$: A randomized key generation algorithm that takes the security parameter $\lambda$ and outputs a symmetric key $k$.
    \item $\Enc(k, m, \vec{h}) \to st$: A randomized encoding algorithm that takes a key $k$, a message $m$, and a history $\Vec{h} = (h_1, \ldots, h_\ell)$ of previously sent stegotext messages, and outputs a stegotext message $st$.
    \item $\Dec(k, \vec{h}, st) \to m/\bot$: A deterministic decoding algorithm that takes a key $k$, a history $\vec{h}$, and a stegotext $st$, and outputs a message $m$ or a rejection symbol $\bot$.
\end{itemize}

\paragraphb{Correctness} 
    A steganographic scheme is correct, if decoding an honestly generated stegotext yields the hidden message with overwhelming probability.
\begin{definition}\label{def:Correctness}
    A steganographic scheme $\Pi = (\KG, \Enc, \Dec)$ is said to be correct if for $\lambda\in \NN$, any message $m$, and any history $\Vec{h} \in \Sigma^\ast$ it holds that the following function is negligible in $\lambda$:
\[  \Pr \left[ m' \neq m  \ : \ \begin{array}{c}
     k \sample \KG(\secparam) \\
     st \sample  \Enc(k,m,\vec{h})\\
     m' \gets \Dec(k, \Vec{h}, st)
\end{array}\right]. \]%
\end{definition}

\paragraphb{Security} Informally, a steganographic scheme is said to be secure if no efficient adversary can detect whether communication over the channel $Ch$ bears a hidden message or not. This is formalized by granting the adversary oracle access to either an encoding oracle $\Enc(k, \cdot, \cdot)$ or to a channel-sampling oracle $\mathcal{O}_{\mathcal{D}}(\cdot, \cdot)$. The oracle $\mathcal{O}$ takes in a message $m$ and a history $\Vec{h}$. It ignores $m$ and replies with a response $\Vec{a} \in \Sigma^\ast$ sampled according to the distribution $\mathcal{D}(\Vec{h})$. An efficient adversary should not be able to tell whether it is given access to  $\Enc(k, \cdot, \cdot)$ or to $\mathcal{O}_{\mathcal{D}}(\cdot, \cdot)$.

\begin{definition}\label{def:CHASecurity}
     A steganographic scheme $\Pi = (\KG, \Enc, \Dec)$ is said to be {\em secure against chosen hidden-text attacks} if for every probabilistic polynomial-time adversary $A$, there exists a negligible function $\nu(\cdot)$ such that  
	\[ \advantage{cha}{D, \Pi} \deq \left|   \Pr \left[  {A}(1^\lambda)^{\Enc(k, \cdot, \cdot)} = 1  \right]  -  	\Pr \left[  {A}(1^\lambda)^{\mathcal{O}_{\mathcal{D}}(\cdot, \cdot)} = 1  \right] \right| \leq \nu(\lambda),  \]%
	for all sufficiently large $\lambda \in \mathbb{N}$, where  the probabilities are taken over $k \sample \KG(\secparam)$, the randomness of $A$ and the randomness of the oracles.
\end{definition}

\subsection{Weak Robustness}

Under the weak threat model, robustness ensures that the hidden message can be successfully recovered even if the stegotext is altered by a function \(f \in \mathcal{F}_{k, \epsilon}\), where \(\mathcal{F}_{k, \epsilon}\) denotes the family of locally consistent tampering functions.

\begin{figure}[h]
    \begin{pcvstack}[center,boxed]
        \procedure[linenumbering]{$\mathbf{WeakRobust}_{\adv, \Pi, \mathcal{F}_{k,\epsilon}}(\secpar)$}{
            (\Vec{h}, m, \state) \sample A(\secparam) \\
            k \sample \KG(\secparam) \\
            st \sample \Enc(k, m, \Vec{h}) \\
            (st', f) \sample A(\state, st) \\
            \pcif f \not\in \mathcal{F}_{k, \epsilon} \pcthen \pcreturn 0 \\
            m' \gets \Dec(k, \Vec{h}, st') \\
            \pcif m \neq m' \pcthen \pcreturn 1 \ \pcelse \pcreturn 0
        }
    \end{pcvstack}
    \caption{The weak robustness security game for a steganographic scheme $\Pi$, an adversary $A$, and a tampering function family $\mathcal{F}_{k, \epsilon}$.}
    \label{fig:weakrobust}
\end{figure}

\begin{definition}[Weak Robustness]\label{def:WeakRobustness}
    A steganographic scheme $\Pi = (\KG, \Enc, \Dec)$ is said to be $\mathcal{F}_{k, \epsilon}${\bf -robust} if for every probabilistic polynomial-time adversary $A$, there exists a negligible function $\nu(\cdot)$ such that:
\[
\advantage{rbst}{A, \Pi, \mathcal{F}_{k, \epsilon}} \deq \Pr \left[ {\bf WeakRobust}_{\adv,\Pi, \mathcal{F}_{k,\epsilon}}(\secpar) = 1 \right] \leq \nu(\lambda),
\]
for all sufficiently large $\lambda \in \mathbb{N}$. Figure~\ref{fig:weakrobust} contains the definition of the weak robustness security game.
\end{definition}

\subsection{Strong Robustness}

Under the strong threat model, robustness ensures that the hidden message can be successfully recovered even if the stegotext is paraphrased, as long as the paraphrased text remains semantically similar to the original. This similarity is measured by the proximity of their embeddings in the embedding space.

\begin{figure}[h]
    \begin{pcvstack}[center,boxed]
        \procedure[linenumbering]{$\mathbf{StrongRobust}_{\adv, \Pi, \delta}(\secpar)$}{
            (\Vec{h}, m, \state) \sample A(\secparam) \\
            k \sample \KG(\secparam) \\
            st \sample \Enc(k, m, \Vec{h}) \\
            st' \sample A(\state, st) \\
            \pcif \text{distance}(E(st), E(st')) > \delta \pcthen \pcreturn 0 \\
            m' \gets \Dec(k, \Vec{h}, st') \\
            \pcif m \neq m' \pcthen \pcreturn 1 \ \pcelse \pcreturn 0
        }
    \end{pcvstack}
    \caption{The strong robustness security game for a steganographic scheme $\Pi$, an adversary $A$, and an embedding distance threshold $\delta$.}
    \label{fig:strongrobust}
\end{figure}

\begin{definition}[Strong Robustness]\label{def:StrongRobustness}
    A steganographic scheme $\Pi = (\KG, \Enc, \Dec)$ is said to be $\delta${\bf -robust} if for every probabilistic polynomial-time adversary $A$, there exists a negligible function $\nu(\cdot)$ such that:
\[
\advantage{rbst}{A, \Pi, \delta} \deq \Pr \left[ {\bf StrongRobust}_{\adv,\Pi,\delta}(\secpar) = 1 \right] \leq \nu(\lambda),
\]
where \(\delta > 0\) is a distance threshold such that \(\text{distance}(E(st), E(st')) \leq \delta\) ensures the embedding of the tampered stegotext $st'$ is close to that of the original $st$. Figure~\ref{fig:strongrobust} contains the definition of the strong robustness security game.
\end{definition}

\FloatBarrier
\section{Constructing Robust Steganographic Schemes}

In this section, we present our schemes for robust steganography relying on LLMs. We design a scheme based on using multiple watermarks to embed information into a covertext that satisfies the conditions needed for security under the weak robustness model and design a scheme based on rejection sampling the hash outputs of text embeddings that satisfies the conditions needed for security under the strong robustness model. We additionally prove their correctness and security for their respective threat models.

\subsection{Watermarking Scheme for Weak Robustness}
\label{sec:watermarkscheme-theory}

Our approach builds on prior work, using techniques that embed a watermark as a primitive into text generation, where the presence of a watermark can be interpreted as encoding a single bit of information. We extend this concept to encode multiple bits, leveraging statistical detection techniques (Appendix~\ref{covertext_length_derivation}) to determine which watermarking keys were used. This enables reliable message recovery for multi-bit messages while ensuring covertexts are indistinguishable from randomly sampled messages in the covertext space.

\subsubsection{Encoding Process.}

To embed a message $m$, we generate text using an LLM while perturbing token probabilities based on $k$ and $m$. For simplicity, assume we encode messages $m \in \{0,1\}^\ell$ of bounded length $\ell$, with the secret key structured as $k = k_1 \| \cdots \| k_\ell$, where each $k_i$ is an independent watermarking key.\footnote{For longer messages, $k$ can be derived using a pseudorandom function $F$, setting $k_i = F_k(i)$.} Conceptually, for a message $m = m_1 m_2 \cdots m_\ell$, the generated text is influenced by the watermarking keys $\{ k_{i} \}_{i : m_i = 1}$ but remains independent of $\{ k_{i} \}_{i : m_i = 0}$.

More formally, during text generation, when sampling the $j$th token $T_j$, a pseudorandom function (PRF) $F$ which outputs a bitstring in $\{0,1\}^N$ (where $N$ is the vocab size of the LLM) is applied to the preceding $c$-gram of tokens:

\[
r = F_k(T_{j-c}, \dots, T_{j-1})
\]

The output $r$ determines how the probability distribution over $T_j$ is perturbed—boosting tokens that correspond to $r_i = 1$ while suppressing tokens that correspond to $r_i = 0$. This ensures that subtle modifications to token likelihoods encode information in a way that is recoverable but not easily detectable without the key $k$.

\subsubsection{Addressing Token Reuse and Detection Risks.}

A challenge in steganography is ensuring that the perturbations do not create detectable statistical artifacts. If $c$ is too small, certain $c$-grams may repeat frequently, revealing that the text distribution has been altered. To mitigate this, we introduce a salting mechanism: both sender and receiver maintain a deterministic state $s$ derived from the conversation history (e.g., the number of exchanged messages). Instead of computing $F_k(T_{j-c}, \dots, T_{j-1})$ directly, we use:

\[
r = F_k(s, T_{j-c}, \dots, T_{j-1})
\]

This ensures that even for the same $c$-gram, different conversations produce different perturbations, reducing the likelihood of detection.

Across different messages, varying the salt value means that repeated $c$-grams are unlikely to receive the same perturbation. Thus, an adversary attempting to detect a persistent watermark by analyzing recurring patterns will face increased difficulty, as different texts will exhibit different perturbations even under identical generation conditions.

\subsubsection{Decoding via Statistical Inference.}

To recover the hidden message, we leverage the statistical properties of the watermarking mechanism. We test whether each key $k_i$ was used to perturb the token distribution. The detection follows the methodology from prior watermarking work~\cite{red_green_watermark}:

\begin{enumerate}
	\item The decoder reconstructs the set of boosted and suppressed tokens for each $k_i$.
	\item It analyzes whether the sampled tokens in the observed text align with the expected distribution shifts for a given $k_i$.
	\item Using statistical tests (e.g., a z-test over token frequency deviations), the decoder determines whether the perturbations associated with $k_i$ are present.
	\item If the statistical test indicates the presence of perturbations from $k_i$, we infer $m_i = 1$; otherwise, we infer $m_i = 0$.
\end{enumerate}

\subsubsection{Robustness Against Targeted Attacks.}

An adversary might attempt to remove a hidden message by selectively disrupting certain bits. To defend against such targeted attacks, we use a key-selection strategy. Rather than deterministically cycling through the watermarking keys $k_i$ for bits $m_i = 1$, we dynamically select a key from this set at each covertext token generation step.

Specifically, we maintain a list of keys:

\[
K_m = \{ k_i \mid m_i = 1 \}
\]

and use a PRF to randomly select a key from this list for each token position $j$:

\[
k_j = K_m[F_k(j) \mod |K_m|].
\]

This selected key $k_j$ is then used in another PRF, salted with the conversation state $s$, to perturb the probability distribution when generating token $T_j$:

\[
r = F_{k_j}(s, T_{j-c}, \dots, T_{j-1}).
\]

This process ensures that the perturbation is unpredictable while still recoverable by the recipient, as different keys are sampled dynamically while maintaining alignment between sender and receiver. This makes it difficult for an attacker to selectively erase individual bits without corrupting the entire message, as key usage is non-trivial to predict.

\subsubsection{Algorithm Overview}
\leavevmode\par
\medskip
Here we give a high-level overview of the encoding and decoding algorithms. Full pseudocode is provided in Algorithm~\ref{alg:perturb} and Algorithm~\ref{alg:encode} for the encoding process, and in Algorithm~\ref{alg:decode} for the decoding process. \\

\noindent \textbf{Encoding Algorithm:}
\begin{enumerate}
	\item Compute salt $s$ from conversation history.
	\item For each token $T_j$:
	\begin{enumerate}
		\item Compute $k_j$
    	\item Compute $r = F_{k_j}(s, T_{j-c},\dots,T_{j-1})$.
    	\item Modify the probability distribution over $T_j$ using $r$.
    	\item Sample $T_j$ from the perturbed distribution.
	\end{enumerate}
	\item Output the generated text as stegotext.
\end{enumerate}

\noindent \textbf{Decoding Algorithm:}
\begin{enumerate}
	\item Compute boosted and suppressed token lists for each possible watermarking key $k_i$.
	\item Count occurrences of boosted and suppressed tokens in the observed text.
	\item Perform statistical tests to determine whether the distribution matches the expected perturbations for $k_i$.
	\item Infer the hidden bit $m_i$ based on statistical confidence.
\end{enumerate}

This watermarking-based approach provides a practical steganographic mechanism for weak robustness. The integration of salting and dynamic key selection weakens adversarial attacks, ensuring that hidden messages remain difficult to detect while remaining recoverable by authorized recipients.

\begin{algorithm}
\caption{Perturb}\label{alg:perturb}
\begin{algorithmic}[1]
\Require A probability distribution $\vec{p}=(p_1,\ldots,p_N)$, a binary vector $\vec{r}\in\{0,1\}^N$, and a perturbation parameter $\delta\in(0,1)$.
\State Let $\mathcal{I} \gets \{\,i\in[N]: p_i\in[2\delta,1-2\delta]\,\}$
\State Let $w \gets |\{\,i\in\mathcal{I} : r_i=1\,\}|$
\State Compute $\delta' \gets \dfrac{\delta\,w}{N-w}$
\ForAll{$j\in\mathcal{I}$}
    \If{$r_j=1$}
       \State Set $p'_j \gets p_j + \delta$
    \Else
       \State Set $p'_j \gets p_j - \delta'$
    \EndIf
\EndFor
\ForAll{$j\in [N]\setminus\mathcal{I}$}
    \State Set $p'_j \gets p_j$
\EndFor
\State \Return $\vec{p'} = (p'_1,\ldots,p'_N)$
\end{algorithmic}
\end{algorithm}

\begin{algorithm}
\caption{Encode}\label{alg:encode}
\begin{algorithmic}[1]
\Require key $k = (k_1,\ldots, k_\ell)$, history $\vec{h}$, message to be encoded $m\in \{0,1\}^\ell$, hardness parameter $\delta$
\State Compute the salt $s$ from the history $\vec{h}$.
\For{$j=1,2,\ldots$}
    \State Sample $i \leftarrow \{ j\in [\ell] \ : \ m_j = 1  \}$ using a PRF
    \State Apply the language model over previous tokens to get a probability distribution $\vec{p}$ over the $t$th token
    \State Compute $\vec{r} \gets F_{k_i}(s,T_{j-c},\ldots, T_{j-1})$, where $T_{j-c},\ldots, T_{j-1}$ are the $c$ prior tokens
    \State Compute the perturbed distribution $\Vec{p'} \gets \Perturb(\vec{p},\Vec{r},\delta)$.
    \State Sample the next token $T_j$ using the distribution $\Vec{p'}$.
\EndFor \\
\Return the generated text as the stegotext $st \gets T_1 T_2 \ldots$
\end{algorithmic}
\end{algorithm}

\begin{algorithm}
\caption{Decode}\label{alg:decode}
\begin{algorithmic}[1]
\Require key $k = (k_1,\ldots, k_\ell)$, history $\vec{h}$, stegotext $ct = T_1 T_2 \ldots$, threshold parameter $z$
\State Compute the salt $s$ from the history $\vec{h}$.
\State Initialize counters $ctr_1,\ldots, ctr_\ell \gets 0$
\For{$j=1,2,\ldots$}
    \For{$i = 1,\ldots, \ell$}
    \State Compute $\vec{r} \gets F_{k_i}(s,T_{j-c},\ldots, T_{j-1})$, where $T_{j-c},\ldots, T_{j-1}$ are the $c$ prior tokens
    \State Let $q \in \{0,1,\cdots,N-1\}$ denote the index of $T_j$ within the vocabulary. If $r_q = 1$, increment $ctr_i \gets ctr_i + 1$
\EndFor 
\EndFor
\State For each $i \in [\ell]$, if $ctr_i \geq z$, set $m_i \gets 1$. Otherwise, set $m_i \gets 0$ 
\State \Return $m = m_1 m_2 \cdots m_\ell$
\end{algorithmic}
\end{algorithm}

\paragraphb{Proofs}
Proofs are given in Appendix~\ref{watermark_theorems}.

\subsection{Embedding Scheme for Strong Robustness}
\label{sec:embedscheme-theory}

To achieve strong robustness, we leverage text embeddings and a locally sensitive hash (LSH) function to encode messages in a way that is resilient to paraphrasing attacks. The core idea is to use rejection sampling to generate stegotext whose embedding hashes to the desired message bits. We assume that the LSH is keyed in the following explanation or that it is used on already encrypted data to ensure confidentiality.

\paragraphb{Main Idea} 
The embedding scheme encodes a message \(m \in \{0,1\}^n\) by breaking it into chunks of \(l\) bits each. For each chunk, rejection sampling is applied during text generation to ensure that the embedding of the generated stegotext maps to the corresponding \(l\)-bit chunk under the LSH. This process continues iteratively until all \(n\) bits of the message are encoded. The decoding process reconstructs \(m\) by computing the LSH of the received text embeddings and concatenating the resulting \(l\)-bit hashes.

\paragraphb{Encoding}
The encoding algorithm ensures that the generated text’s embedding satisfies the desired hash values. The process is as follows:
\begin{enumerate}
    \item \textbf{Initialization:} Break the message \(m\) into \(\lceil n / l \rceil\) chunks, \(m_1, m_2, \ldots, m_{\lceil n / l \rceil}\), where each chunk \(m_i \in \{0,1\}^l\).
    \item \textbf{Chunk-wise Generation:} For each chunk \(m_i\):
    \begin{itemize}
        \item Generate a candidate text \(x_i\) using the LLM.
        \item Compute the embedding \(E(x_i)\) of the generated text.
        \item Apply the LSH function \(H(E(x_i))\) to the embedding.
        \item If \(H(E(x_i)) = m_i\), accept \(x_i\) and move to the next chunk. Otherwise, repeat the process with a new candidate text.
    \end{itemize}
    \item \textbf{Output:} Concatenate the accepted texts \(x_1, x_2, \ldots\) to form the stegotext \(st\).
\end{enumerate}

\paragraphb{Decoding}
The decoding algorithm reconstructs the hidden message \(m\) from the received stegotext \(st\):
\begin{enumerate}
    \item \textbf{Chunk-wise Decoding:} Break the stegotext \(st\) into the same chunks \(x_1, x_2, \ldots\) used during encoding.
    \item \textbf{Hash Extraction:} For each chunk \(x_i\), compute its embedding \(E(x_i)\) and apply the LSH function \(H(E(x_i))\) to extract the \(l\)-bit hash.
    \item \textbf{Message Reconstruction:} Concatenate the hashes of all chunks to reconstruct \(m\).
\end{enumerate}

\paragraphb{Robustness to Paraphrasing}
The LSH ensures that semantically equivalent paraphrases map to the same hash with high probability, as long as the embeddings remain within a certain proximity. This guarantees that the system tolerates adversarial paraphrasing, provided the changes do not cause significant shifts in the embedding space.

\paragraphb{Advantages}
This approach offers several benefits:
\begin{itemize}
    \item \textbf{Robustness:} The use of LSH and rejection sampling ensures resilience to paraphrasing, aligning with the strong robustness definition.
    \item \textbf{Flexibility:} The bit-length \(l\) of the hash can be adjusted to balance efficiency and the reliability of the chosen LSH.
    \item \textbf{Scalability:} The rejection sampling process ensures that any \(n\)-bit message can be encoded reliably, regardless of its length.
\end{itemize}

\paragraphb{Challenges}
The main challenges of this approach lie in constructing LSH functions with suitable accuracy and the potential computational cost of rejection sampling.

\paragraphb{Proofs}
Proofs are given in Appendix~\ref{embedding_theorems}.
\FloatBarrier
\section{Implementation and Measurement}
\label{implementation-and-measurement}

In this section, we detail our publicly available reference implementation of the two proposed steganographic systems: the watermarking-based scheme (for weak robustness) and the embedding-based scheme (for strong robustness). Both systems are fully open-source and available at \url{https://github.com/NeilAPerry/robust_steganography}. The code is designed to be modular, allowing developers to:
\begin{itemize}
	\item Integrate new covertext sources (e.g., GPT-2, custom Hugging Face transformers), new perturbing functions for sampled tokens, and new PRFs for watermarking.
	\item Integrate new embedding models, new covertext sources, new LSH functions, new error correcting codes, new data encoders, and train new models for particular speech patterns for the embedding-based scheme.
	\item Configure attack models to test new and existing options.
\end{itemize}

We also describe our automated attack suite, which simulates real-world tampering scenarios. Finally, we provide preliminary performance measurements, illustrating each system’s performance and resilience to various attacks.

\subsection{Implementation Overview}
Our implementation is structured as follows:

\begin{itemize}
	\item \textbf{Covertext Generation:}
	\begin{itemize}
    	\item \emph{Watermarking Scheme:} By default, we include GPT-2 and a nanoGPT Shakespeare model \cite{nanogpt} for demonstration. Any GPT-style transformer exposing next-token probabilities can be added.
    	\item \emph{Embedding Scheme:} We provide example scripts using ``gpt-4o-mini'' for text generation, but the system can interface with any OpenAI API-compatible LLM.
	\end{itemize}

	\item \textbf{Message Encoding \& Decoding:}
	\begin{itemize}
    	\item \emph{Watermarking Scheme:} Perturbs token distributions (Section~\ref{sec:watermarkscheme-theory}) to embed bits, with a salting mechanism reducing repeated $c$-gram patterns across multiple messages.
    	\item \emph{Embedding Scheme:} Uses a Locality-Sensitive Hash (LSH) to ensure the text’s embedding aligns with target bits (Section~\ref{sec:embedscheme-theory}), employing rejection sampling until a suitable candidate is found.
	\end{itemize}

	\item \textbf{Error Correction and Data Encoding:}
	\begin{itemize}
    	\item \emph{Repetition Codes} and \emph{Convolutional Codes} help mitigate lower-than-ideal raw accuracy in the LSH. When stronger error correcting codes (ECC) are used, encoded messages grow in length, balancing accuracy vs. covertext length.
    	\item Data-encoding utilities allow users to embed text and ciphertext blocks under several encoding schemes.
	\end{itemize}

	\item \textbf{LSH:}
	\begin{itemize}
    	\item \emph{Random Projection LSH} as a baseline.
    	\item \emph{PCA-based LSH} (Section~\ref{sec:lsh-design}), trained on corpora such as Enron emails and Wikipedia articles to provide domain-tailored projections.
	\end{itemize}

	\item \textbf{Attack Simulator:} Automates realistic tampering methods, including $n$-gram shuffling, synonym substitution, and AI-based paraphrasing. Users can set parameters\footnote{The attack system includes three main modes: $n$-gram shuffle which breaks the text into $n$-grams and shuffles them; synonym swaps which replaces words with their synonyms; and paraphrase which has an LLM paraphrase the text. These include local and global options, which runs the attack contained within 1 sentence chunks or the entire text respectively. One can also specify the tokenizer, parameters such as $n$, the proportion of the text that should be attacked, the source/method of synonyms, the model to use for paraphrasing, and set the temperature of the paraphrasing model.} (e.g., fraction of replaced text and local or global options) to gauge robust message recovery under various threat levels.
\end{itemize}

This modular design allows new cryptographic primitives, language models, LSH variants, or ECC schemes to be used with minimal changes to the rest of the system.

\subsection{Watermarking-Based System}
\label{sec:watermarkscheme-implementation}

Our watermarking-based system (Section~\ref{sec:watermarkscheme-theory}) encodes bits into covertext by perturbing token probability distributions. Key components include:

\begin{itemize}
	\item \textbf{Token Perturbation:}
	We implement two sample perturbation functions—one operating on undetectable subtle shifts (similar to and Kirchenbauer~et~al.'s and Aaronson's watermarking schemes \cite{red_green_watermark,aaronson_talk}) and a ``sharp’’ perturbation that zeros out token probabilities for debugging and testing. Both extract the raw token probabilities from the chosen LLM and apply a keyed pseudorandom function (PRF) to ``boost’’ or ``suppress’’ selected tokens.

	\item \textbf{PRF Integration:}
	Our default setup uses HMAC-based PRFs, but the design allows any keyed PRF to be provided by the user and an example of adding an AES based option. PRF outputs guide which token probabilities get altered; users can also configure a perturbation strength \(\delta\).

	\item \textbf{Length Estimation:}
	Before encoding, the system calculates required covertext length requirements using empirical statistics (Appendix~\ref{covertext_length_derivation}) to ensure the hidden message is recoverable at a specified probability \(\epsilon\).

	\item \textbf{History Salt:}
	Each new message is salted with a function of the current conversation state so repeated $c$-grams across multiple stegotexts do not yield identical perturbations.
\end{itemize}

\paragraph{Usage and Customization.}
Users configure model selection, PRF parameters, $c$-gram size, and perturbation parameters. The system produces a final stegotext.

\subsection{Embedding-Based System}
\label{sec:embedscheme-implementation}

In the embedding-based system (Section~\ref{sec:embedscheme-theory}), candidate texts are sampled and then accepted or rejected based on their LSH output. The main features are:

\begin{itemize}
	\item \textbf{Flexible Embedding Models:}
	Our examples use a “gpt-4o-mini” model for text generation and \(\texttt{text-embedding-3-large}\) for embedding extraction, though other OpenAI API endpoints can be substituted.

	\item \textbf{Rejection Sampling:}
	If the text’s embedding fails to yield the correct LSH bits, the sample is rejected and a new one is generated.

	\item \textbf{Error Correction:}
	The built-in repetition and convolutional codes help protect hidden bits from hashing errors. With multi-bit hashes, strong ECC often becomes crucial to maintain acceptable decoding success rates.
\end{itemize}

\paragraph{Cost Tracking.}
For the embedding-based system using OpenAI-based text generation, we empirically find the cost to be less than 1¢ per hidden byte. See Section~\ref{sec:embedding_measurements} for a full analysis.

\subsubsection{LSH Design}
\label{sec:lsh-design}

We provide two LSH variants:

\begin{itemize}
	\item \textbf{Random Projection:}
	A baseline approach that projects embeddings onto random hyperplanes to generate bits.

	\item \textbf{PCA-based Hash:}
	\begin{itemize}
    	\item \emph{Training:} A PCA model is trained on text embeddings from a target domain (e.g., Enron “corporate speech,” Wikipedia articles).
    	\item \emph{Hashing:} Each principal component forms a dimension in the transformed embedding space, thresholded to yield one bit. Combining multiple dimensions produces an $l$-bit hash.
    \end{itemize}
\end{itemize}

To measure performance we ran two tests. In the first, we sampled the next message from an initial history 1000 times, paraphrased each sampled message, and compared the accuracy of both functions (they are the same if both messages hash to the same value and different otherwise). The Random Projection had an accuracy of 87.6\% and the PCA Hash had an accuracy of 91.6\%. In the second, we conintuously sampled 1000 messages using the initial history plus the previous sampled messages as the history, paraphrased each sampled message, and compared the accuracy of both functions. The Random Projection had an accuracy of 77.1\% and the PCA Hash had an accuracy of 83.8\%\footnote{These were particularly long samples, and for shorter messages we observe rates of 93\% for the Random Projection Hash and 95-97\% for the PCA Hash.}.

\subsubsection{Simulator for Development and Testing}
A separate simulator tool enables evaluation of different embeddings or the effects of new LSH algorithms by using an oracle. This allows experimentation with minimal API usage, particularly when tuning parameters.

\subsection{Attack Simulator}
\label{sec:attack-simulator}

To assess robustness, we provide an attack suite that modifies the stegotext in realistic ways:

\begin{itemize}
	\item \textbf{$n$-Gram Shuffling:}
	Randomly shuffles small text segments, measuring how the watermarking scheme and embedding-based approach handle attacks such as reordering sections of text.

	\item \textbf{Synonym Substitution:}
	Replaces words with synonyms. Users can tune the fraction of tokens replaced, balancing realism and severity.

	\item \textbf{AI Paraphrasing:}
	Uses an LLM to paraphrase sentences (local) or larger blocks of text (global), reflecting scenarios in which adversaries attempt to preserve semantic meaning while altering how the message is phrased.
\end{itemize}

After each attack, the stegotext is fed into the decoding function to measure recovered message rates.

\subsection{Performance Evaluation}

We conducted experiments to measure each system’s performance, covertext overhead, and robustness against attacks. Our primary metrics include frequency of recurring n-grams within a text, how parameters affect covertext length, message recovery accuracy, rejection sampling efficiency, embedding drift under attacks, and the cost per hidden bit when using paid APIs.

\subsubsection{Watermarking-Based System}

\paragraph{Recurring c-grams}
We measure how frequently $c$-grams occur in generated text as $c$ increases (Figure~\ref{fig:recurring_c_grams}). Intuitively, higher values of $c$ repeat less frequently, leaving little to be learned by adversaries. All values $c > 2$ occur at low rates in covertexts of 100 tokens.

\begin{figure}[h]
    \centering
    \includegraphics[width=0.45\linewidth]{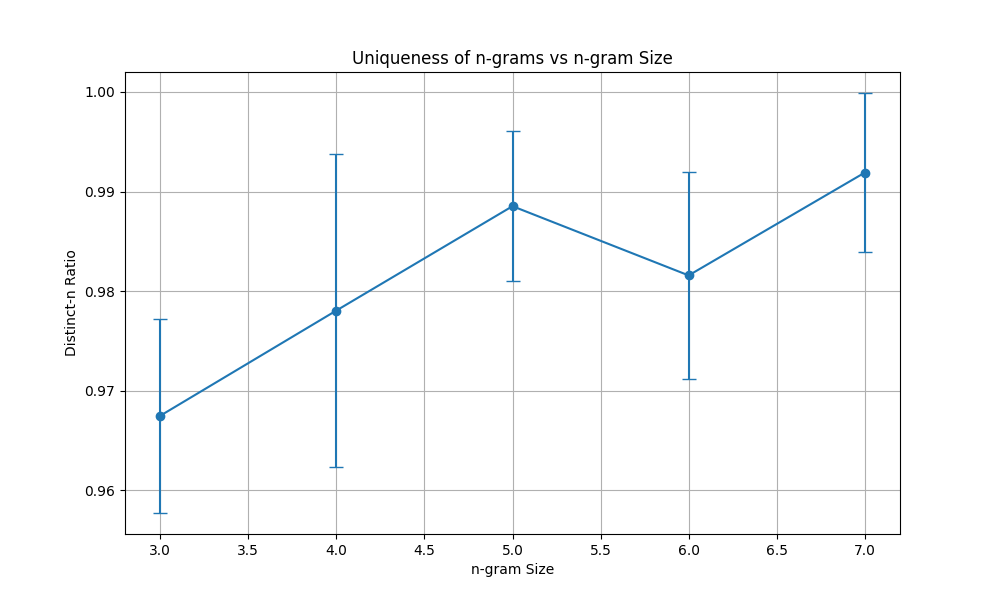}
    \caption{Frequency of recurring $c$-grams is measured in 100 samples of text containing 100 tokens sampled from GPT-2. As the size of $c$ increases, duplicates are less likely to be found.}
    \label{fig:recurring_c_grams}
\end{figure}

\paragraph{$\delta$ Effect on Covertext Length}
We measure how varying $\delta$ impacts covertext length (Figure~\ref{fig:delta_covertext_effect}). We find that as $\delta$ increases, the required covertext length decreases. The trade-off lies in the fact that larger values of delta disrupt the ``naturalness'' of generated text, with high values violating the stealthyness of the scheme.

\begin{figure}[h]
    \centering
    \includegraphics[width=0.45\linewidth]{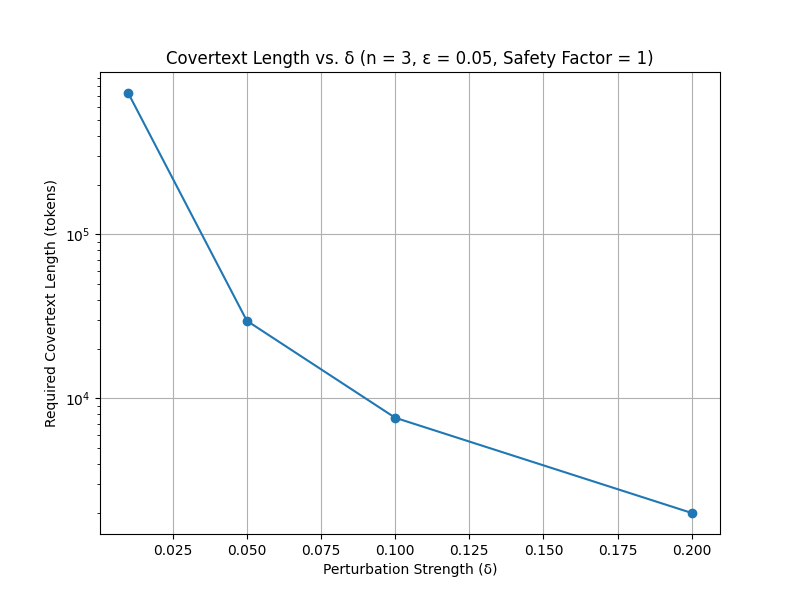}
    \caption{$\delta$'s effect on covertext length. Keeping a fixed $n=3$ and safety factor of 1, we measure the necessary length of the covertext to successfully hide messages with 95\% accuracy.}
    \label{fig:delta_covertext_effect}
\end{figure}

\subsubsection{Embedding-Based System}
\label{sec:embedding_measurements}

\paragraph{Rejection Sampling Efficiency.}
We measure how many samples (on average) are needed to find a contextually relevant message whose embedding hashes to the desired value (Figure~\ref{fig:embedding_rejection_sampling}). We test this for both the Random LSH and the PCA LSH for 1 bit outputs and simulate trials for multi-bit outputs (see Figure~\ref{fig:embedding_rejection_sampling}). This was done over 100 trials and we did not cap the amount of attempts. We found that the Random LSH took 2.82 samples on average over the 100 trials. This differs from the expected $2^n$ hashes due to the next likely tokens being biased to one side of the random cut. The PCA LSH took 2.03 samples on average over 100 trials. The corporate prompt can be found at Appendix~\ref{pca_prompt}. This corresponds to corporate emails from the Enron dataset and was specifically made for the PCA LSH.

\begin{figure}[h]
    \centering
    \includegraphics[width=0.45\linewidth]{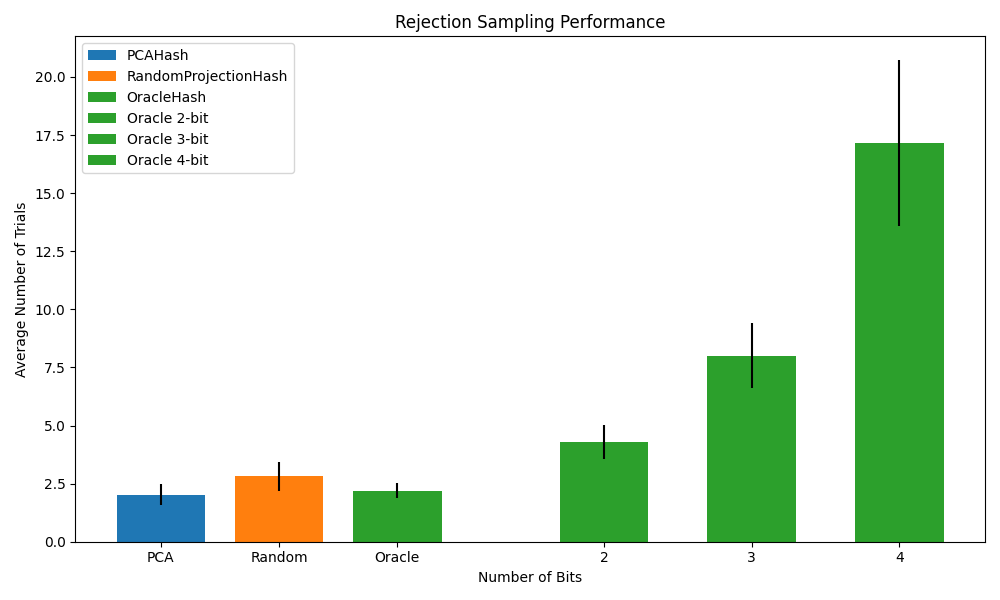}
    \caption{Measurements of embeddings rejection sampling. The Random and PCA LSH functions hash to 1 bit and used a conversational and corporate conversation mode respectively. The corporate mode is designed for the PCA LSH. LSH functions that output more than 1 bit were simulated via the Oracle LSH.}
    \label{fig:embedding_rejection_sampling}
\end{figure}

\paragraph{Embedding Drift.}
We measure how various attacks impact the similarity (Euclidean distance and cosine similarity) of embeddings before and after the attack (Figure~\ref{fig:embeddings_attack_error}). This helps inform which LSH functions can resist certain attack types. As expected, increasing the percentage of the text that is attacked creates further distance between the pair of text-- although n-gram shuffle attacks and paraphrase attacks remain close even at high percentages. Surprisingly, the synonym attack is what causes large changes. This is due to nltk's wordnet being used for the attack and making overly aggressive substitutions. These substitutions include choices that are not grammatically correct, such as changing ``not'' to ``non'' or words that change the meaning of sentences. Therefore, many synonym attacks fall outside the scope of the strong robustness threat model as they are not changes that the adversary would actually want to make and are therefore not a source of concern.

\begin{figure}[h]
    \centering
    \begin{subfigure}[b]{0.45\textwidth}
        \centering
        \includegraphics[width=\textwidth]{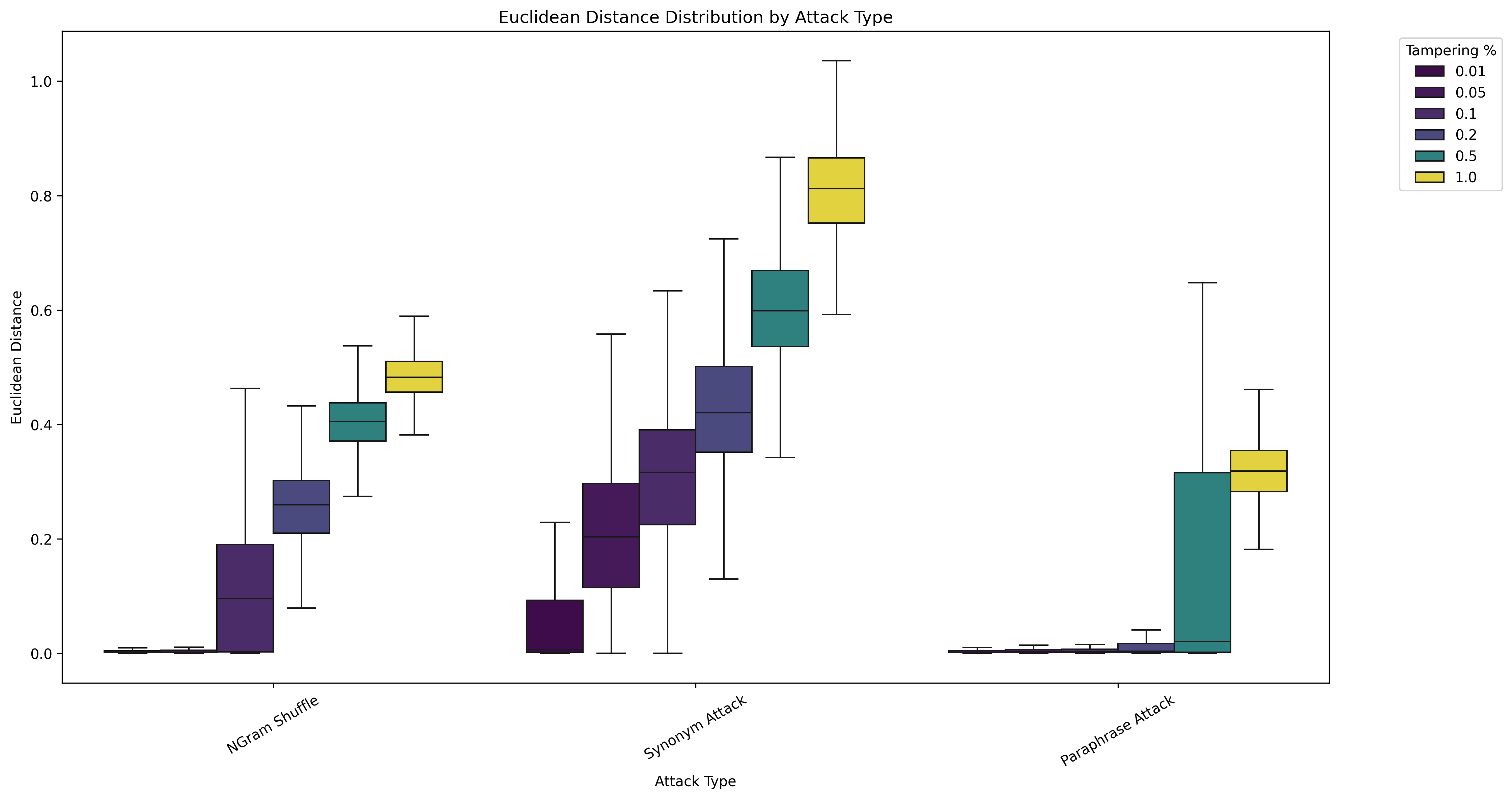}
        \caption{Euclidean distance for various attack categories and percentage of text attacked.}
        \label{fig:embeddings_attack_error_euclidean}
    \end{subfigure}
    \hfill
    \begin{subfigure}[b]{0.45\textwidth}
        \centering
        \includegraphics[width=\textwidth]{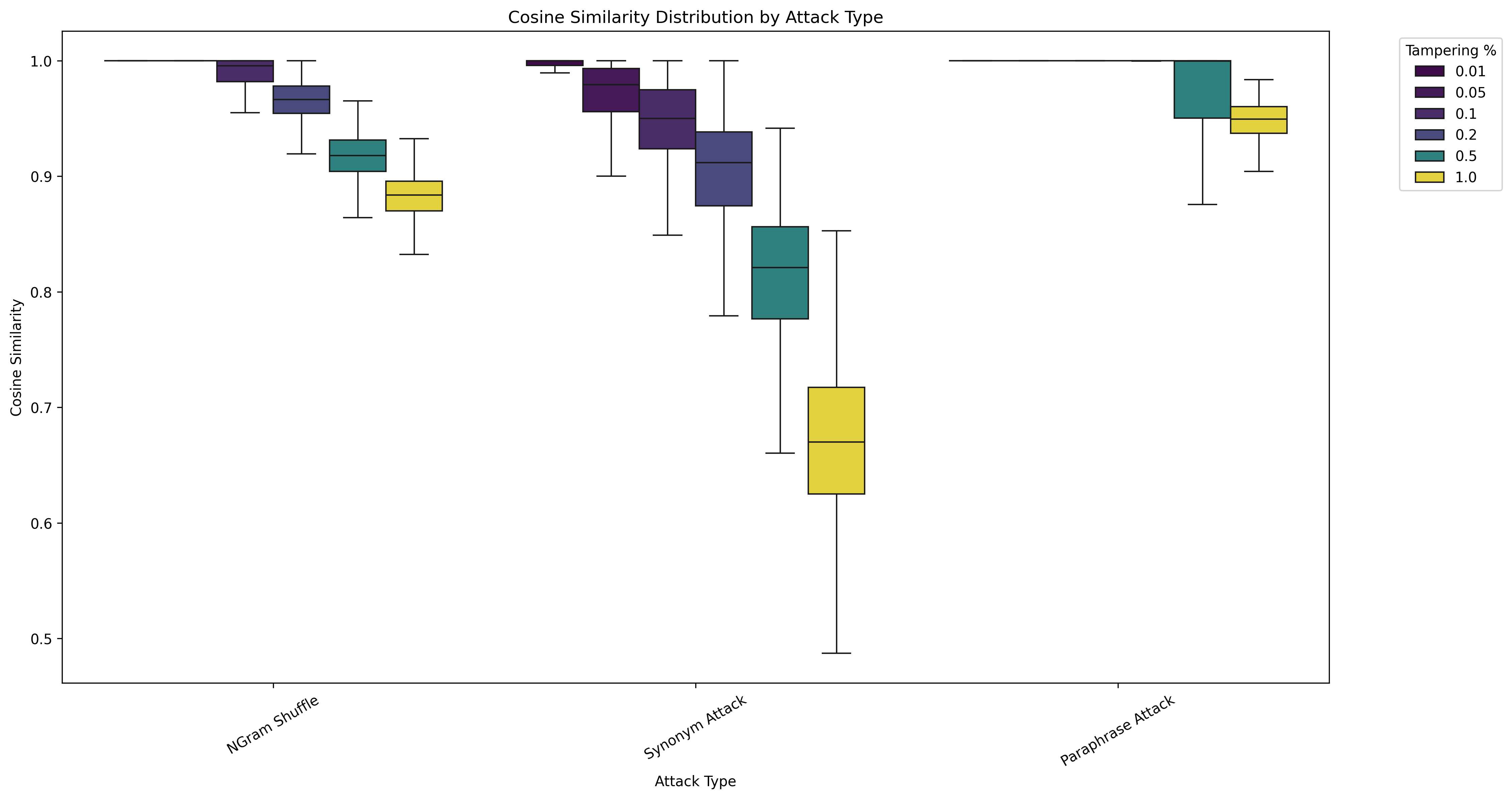}
        \caption{Cosine similarity for various attack categories and percentage of text attacked.}
        \label{fig:embeddings_attack_error_cosine}
    \end{subfigure}
    \caption{}
    \label{fig:embeddings_attack_error}
\end{figure}

\paragraph{Cost per Hidden Bit.}
Here we first derive the cost of hiding an $n$ bit message and then empirically test the cost with OpenAI's API using \texttt{gpt-4o-mini} and \texttt{text-embedding-3-large}. We define the variables as follows:
\begin{itemize}
    \item $n$: Total number of bits to hide.
    \item $h$: Number of bits hidden per successful covertext chunk (i.e., the hash output size).
    \item $c$: Average number of queries needed per successful chunk (for a uniform hash, $c \approx 2^h$).
    \item $W$: Number of input tokens (context window tokens) per query.
    \item $T$: Number of output tokens (covertext chunk tokens) per query.
    \item $p_{\text{in}}$: Cost per input token.
    \item $p_{\text{out}}$: Cost per output token.
\end{itemize}

Each successful chunk hides $h$ bits, so the number of successful chunks required is
\[
\text{Successful Chunks} = \frac{n}{h}.
\]
Since each chunk requires $c$ queries on average, the total number of queries is
\[
\text{Total Queries} = \frac{n}{h} \cdot c.
\]
The cost per query is given by
\[
\text{Cost per Query} = W \cdot p_{\text{in}} + T \cdot p_{\text{out}}.
\]
Thus, the total cost is
\[
\boxed{
\text{Total Cost} = \frac{n}{h} \cdot c \cdot \left( W \cdot p_{\text{in}} + T \cdot p_{\text{out}} \right).
}
\]
For a uniform hash where $c \approx 2^h$, this can be approximated as
\[
\text{Total Cost} \approx \frac{n}{h} \cdot 2^h \cdot \left( W \cdot p_{\text{in}} + T \cdot p_{\text{out}} \right).
\]
We empirically measure the cost of hiding a byte of data in USD by measuring the difference in API costs on the OpenAI usage dashboard before and after each trial. We hide 1 byte of data for both the random projection and PCA-based LSH. Additionally, we set a cutoff of 10 cents per attempt, as the system may enter a low-entropy state where the next continuations are overwhelmingly likely to all hash to a single output. When this happens, we choose to sample an incorrect choice and have error correction fix this on recovery. Table~\ref{table:lsh_cost_comparison} contains the results. We observe that PCA has an overall lower cost when entropy is high enough, but low-entropy situations cause the total cost to be higher. Given that the PCA LSH has a higher accuracy than the Random LSH (Section~\ref{sec:lsh-design}), the optimal choice depends several factors including, length of message, covertext distribution, bandwidth, financial and covertext budgets, and choice in error correcting code.

\begin{table}[h]
    \centering
    \begin{tabular}{|l|c|c|c|}
        \hline
        \textbf{LSH Scheme} & \textbf{Success Rate (\%)} & \textbf{Avg Cost (Success) (¢/Byte)} & \textbf{Avg Cost (Total) (¢/Byte)} \\
        \hline
        Random & 100\% & 0.7 & 0.7 \\
        PCA & 95\% & 0.35  & 0.85 \\
        \hline
    \end{tabular}
    \vspace{1em}
    \caption{Success rate and cost comparison of different LSH schemes for our steganography system taken over 20 trials of hiding 1 byte. Due to low-entropy situations, the next continuation may overwhelmingly be likely to hash to the same value. We therefore cap the cost of hiding a byte at 10¢. The system then sample the wrong value and relies on error correction during retrieval. "Avg Cost (Success)" is the average cost per successfully hidden byte, ignoring failures. "Avg Cost (Total)" includes failures at 10¢ per attempt in the cost.}
    \label{table:lsh_cost_comparison}
\end{table}

\subsubsection{Comparison of the Two Systems}
\label{sec:comparison}

We provide a direct comparison between the watermarking-based and embedding-based approaches, focusing on covertext overhead and resistance to attacks. \\

\textbf{Covertext Overhead:} Plotting hidden message length vs. covertext length to illustrate efficiency differences. Figure~\ref{fig:token_efficiency} compares the Watermarking and Embedding system for various parameters. The embedding system uses substantially less tokens than the watermarking system for all configurations. This includes configurations of the Watermarking system with high levels of delta ($\delta \geq 0.1)$ which can significantly degrades text quality. \\

\begin{figure}[h]
    \centering
    \includegraphics[width=0.45\linewidth]{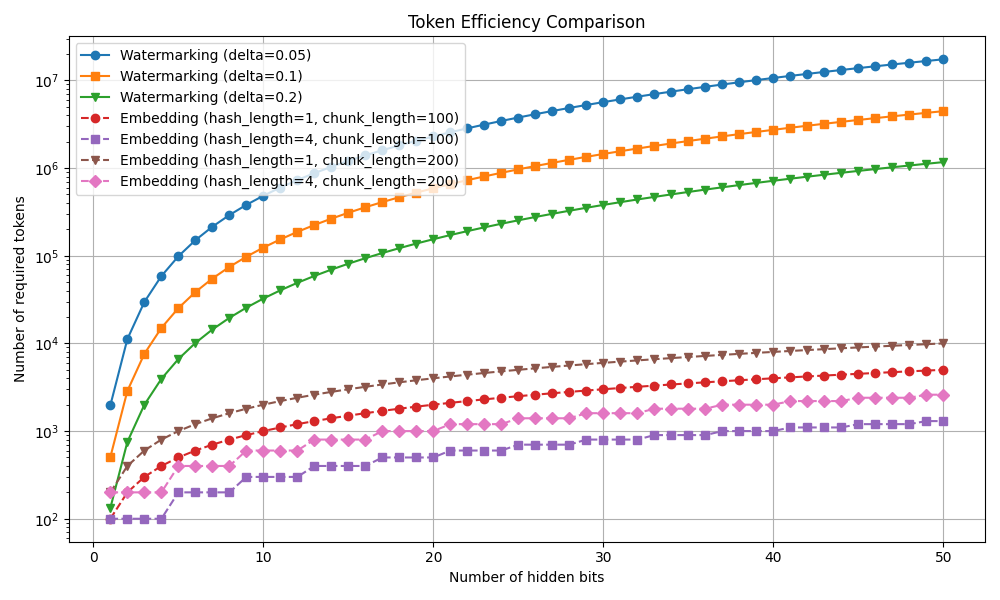}
    \caption{Covertext Overhead. The x-axis contains the number of bits being hidden within the covertext and the y-axis contains the required length of the covertext measured in tokens to hide those bits. The watermarking schemes require orders of magnitude more tokens to hide the same length message as the embedding scheme, even for high levels of delta that would not be used in practice.}
    \label{fig:token_efficiency}
\end{figure}

\textbf{Resistance to Attack:} Contrasting message recovery under varying levels of $n$-gram shuffling, synonym substitutions, and AI paraphrasing. Figure~\ref{fig:watermark_attack_reistance} contains the results of attacking various texts generated via the watermarking system and Figure~\ref{fig:embeddings_attack_reistance} contains the results of attacking various texts generated via the embedding system. For the watermarking system, attacks were run in a local mode that contained modifications to individual sentences and a global mode that attacked the entire text at once. These attacks are in some way an upper bound and can be viewed as overestimating an adversaries capabilities. For example, synonyms are chosen via nltk's wordnet and can be inappropriate, replacing words with substitution that change tense or no longer have the same meaning (i.e. replacing not with non). n-gram shuffles can create completely nonsensical sentences (i.e. changing ``The white cat jumped on the broken couch.'' to ``the broken couch cat jumped on The white.''). This becomes particularly extreme in global mode, where n-grams at the end of long texts can be shuffled together with n-grams at the beginning of a text. Therefore, local attacks may be more realistic in this sense. Additionally, local attacks serve as a more direct comparison between the watermarking and embedding systems. This is due to the challenge from their different threat models and the key difference that hidden bits are ``spread out'' throughout the entire text in the watermarking case compared to hidden bits being ``clumped together'' in discrete sections of the text in the embedding case. These different strengths and weaknesses make local attacks on watermarked text a more appropriate point of comparison to attacks on segments of the embedded text.

The watermarking system withstands n-gram shuffling attacks well up to 20\% of the text being tampered. Local tampering retains a perfect recovery, while global tampering drops to a 90\% recovery rate. It handles synonym and paraphrasing attacks up to 5\% tampering before quickly declining in both the local and global cases.

The embedding system withstands n-gram shuffling and paraphrasing attacks perfectly at all tested level ranging from 1\% to 50\% but declines slightly for synonym attacks down to 95\% recovery when half of the message blocks are attacked. This is partially from non-perfect agreement in LSH accuracies and some synonym substitutions being inapropriate and changing the meaning of text segments.

Overall, the embedding system substantially outperforms the watermarking system, both in the covertext overhead and its resistance to attacks.

\begin{figure}[h]
    \centering
    \begin{subfigure}[b]{0.45\textwidth}
        \centering
        \includegraphics[width=\textwidth]{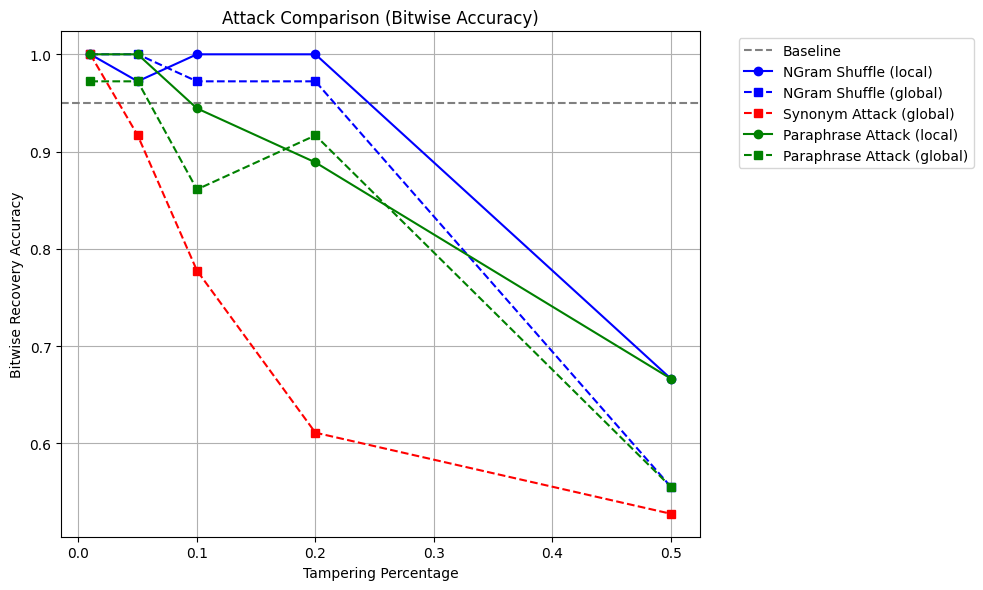}
        \caption{We measure the percentage of bits that are correctly recovered. For example, if [1, 0, 1] was hidden and [0, 0, 1] was recovered, this would contribute a score of 0.67 for that message + attack data point. \\}
        \label{fig:attack_bitwise}
    \end{subfigure}
    \hfill
    \begin{subfigure}[b]{0.45\textwidth}
        \centering
        \includegraphics[width=\textwidth]{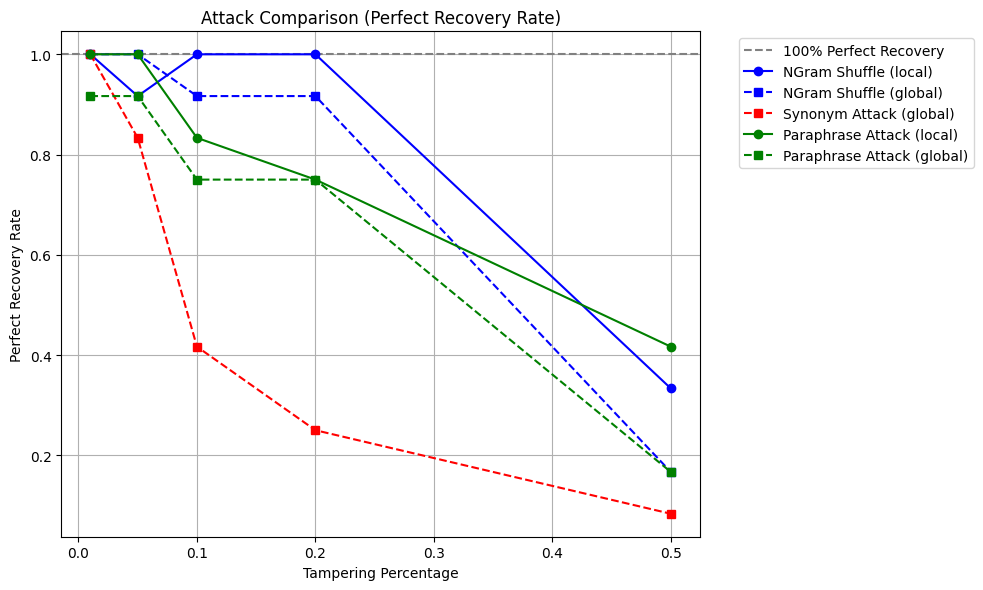}
        \caption{We measure whether the bits that are recovered perfectly match the hidden message. For example, if [1, 0, 1] was hidden and [0, 0, 1] was recovered, this would contribute a score of 0 for that message + attack data point.}
        \label{fig:attack_perfect}
    \end{subfigure}
    \caption{We run n-gram shuffle, synonym, and paraphrasing attacks against 12 stegotexts containing 3 bit messages created via the watermarking scheme for various modes and portions of the text being altered.}
    \label{fig:watermark_attack_reistance}
\end{figure}

\begin{figure}[h]
    \centering
    \begin{subfigure}[b]{0.45\textwidth}
        \centering
        \includegraphics[width=\textwidth]{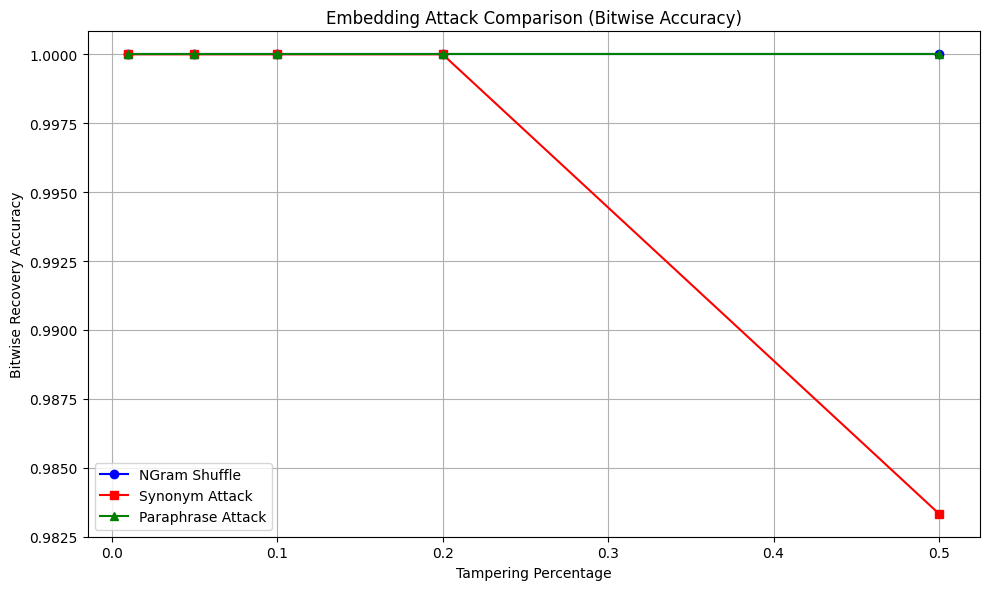}
        \caption{We measure the percentage of bits that are correctly recovered. For example, if [1, 0, 1] was hidden and [0, 0, 1] was recovered, this would contribute a score of 0.67 for that message + attack data point. \\}
        \label{fig:embeddings_attack_bitwise}
    \end{subfigure}
    \hfill
    \begin{subfigure}[b]{0.45\textwidth}
        \centering
        \includegraphics[width=\textwidth]{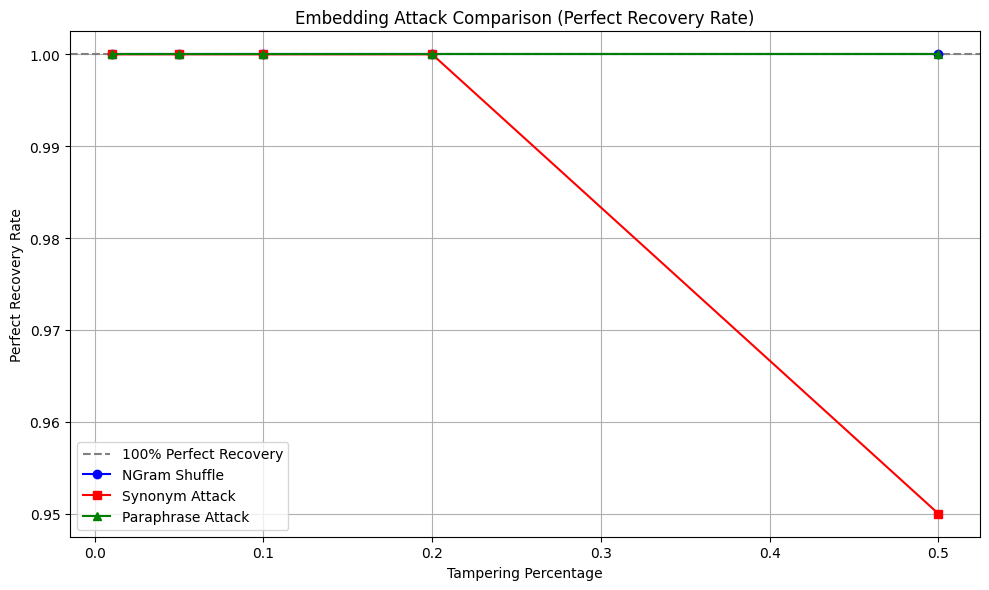}
        \caption{We measure whether the bits that are recovered perfectly match the hidden message. For example, if [1, 0, 1] was hidden and [0, 0, 1] was recovered, this would contribute a score of 0 for that message + attack data point.}
        \label{fig:embeddings_attack_perfect}
    \end{subfigure}
    \caption{We run n-gram shuffle, synonym, and paraphrasing attacks against 20 stegotexts containing 3 bit messages created via the embedding scheme. Attacks are run on each message encoding $hash\_output$ bits (closer to local mode) at various probabilities. The embedding system resists these attacks at a far stronger rate than the watermarking system, but does not possess the property of ``spreading'' the hidden bits throughout the entire text.}
    \label{fig:embeddings_attack_reistance}
\end{figure}

\subsection{Summary of Findings}

Our evaluation demonstrates that the embedding system significantly outperforms the watermarking system in both efficiency and robustness to adversarial attacks, while operating at a reasonable cost.

\textbf{Covertext Efficiency}: The embedding-based approach requires substantially fewer tokens to encode the same number of hidden bits compared to the watermarking-based approach. Even with high perturbation levels ($\delta \ge 0.1$) that degrade text quality, the watermarking system remains significantly less efficient.

\textbf{Resistance to Attacks}: The watermarking system maintains strong recovery under minor modifications (e.g., $\le$20\% n-gram shuffling) but quickly degrades when subjected to synonym substitutions or paraphrasing attacks.
The embedding system demonstrates near-perfect robustness to n-gram shuffling and paraphrasing attacks. However, synonym substitution causes slight degradation (95\% recovery at 50\% attack rate), though this is partly due to inappropriate synonym choices that lay outside of our threat model rather than inherent system weaknesses.

\textbf{Rejection Sampling and Cost}: The PCA LSH achieves higher accuracy than the Random Projection LSH (91.6\% vs. 87.6\% in first-message sampling tests).
Rejection sampling for the embedding system requires fewer iterations for the PCA LSH (2.03 samples on average) compared to the Random Projection LSH (2.82 samples).
The PCA LSH achieves lower cost per successful hidden byte (0.35¢) but is more affected by low-entropy situations, leading to higher total cost when failures are included (0.85¢/byte vs. 0.7¢/byte for the random projection).

\textbf{Practical Considerations}: The embedding system is preferred when strong robustness is required against adversarial paraphrasing. The watermarking system, despite being more fragile, might still be useful in scenarios where minor modifications are expected, and message length constraints are less strict. One key advantage of the watermarking system is that it hides hidden bits in a manner that is diffused throughout the entire text, where no particular section is tied to a particular bit. The embedding system lacks this feature.

Overall, our results highlight that embedding-based steganography provides a more practical and resilient method for covert communication, particularly in environments where adversaries employ sophisticated paraphrasing attacks.

\FloatBarrier
\section{Applications}

In this section, we outline several possible use cases for this type of system.

\subsection{Communication Between Citizens Planning Protests}
Several works have explored communication in unusual settings, such as environments with no internet access \cite{inyangson2024amigo,pradeep2022moby,bienstock2023asmesh,lerner2016rangzen,perry2022strong}, environments where people must use government mandated encryption where the government has access to the corresponding secret keys \cite{horel2018subvert}, or environments where the adversary can use algorithm-substitution attacks to change users' encryption schemes without users' knowing \cite{bellare2014security}. This work is suited to the second environment, where the adversary can essentially view all plaintext sent by users. This work further expands Horel et al. allowing the covertext to remain as English text, therefore not raising suspicion as easily and is robust to rerandomization (i.e. paraphrasing) attacks. This defeats the main attack adversaries had against similar systems, allowing for improved security for users.

\subsection{Border Crossings \& Cloud Storage}
Tyagi~et~al. \cite{tyagi2018burnbox} explores another unusual area with users with sensative needs-- dissidents and journalists at border crossings. These users possess highly sensitive data and face dire consequences if apprehended or captured. Tyagi~et~al. provide a solution through a new file system that provides self-revocable encryption, allowing users to temporarily revoke access to select sensitive files while leaving ``normal looking'' data intact on their local machine or in the cloud. Our system can be used for a similar purpose by storing sensitive files in collections of covertext. For example, sensitive information can be encoded into a series of short stories or poems. When searched, the government would see the innocent files saved to the device or cloud, along with collections of innocuous writing-- which does not seem out of place for someone whose primary job is writing. After the inspection, the hidden data within the poems and short stories can be recovered, restoring the original files. One added benefit is that there is no need to use special storage systems like BurnBox, which may in and of itself raise suspicion. For many realistic adversaries, this may be sufficient to result in punitive action.
\FloatBarrier
\section{Conclusion}
We introduced formal definitions of weak and strong robust steganography, along with threat models that justify their necessity. To address these challenges, we designed and implemented steganographic schemes that enable users to hide arbitrary secret messages within English covertext generated by large language models. These schemes ensure message recoverability even under adversarial paraphrasing and rewording attacks, significantly advancing the robustness of text-based steganography. Our work provides both theoretical guarantees and practical implementations, demonstrating the feasibility of robust steganographic communication in real-world settings. By open-sourcing our stegosystems and datasets, we aim to facilitate adoption by individuals in need while encouraging further research on key components, such as embedding-based locally sensitive hash functions. We believe this research lays a foundation for more resilient and accessible steganographic methods, contributing to secure and censorship-resistant communication.

\paragraph{\textbf{Acknowledgments.}}
We would like to thank Gabriel Poesia for helpful conversations on the design of the schemes in this work. This work was funded by NSF, DARPA, and the Simons Foundation. Opinions, findings, and conclusions or recommendations expressed in this material are those of the authors and do not necessarily reflect the views of DARPA.

\bibliographystyle{abbrv}
\bibliography{main,abbrev2,crypto}

\appendix

\section*{Appendix}

\section{Multiple Watermark Detection}
\label{covertext_length_derivation}

\subsection*{Setup and Notation}

\begin{itemize}
  \item $n$: Total number of bits in the hidden message $(m_1, \dots, m_n) \in \{0,1\}^n$.
  \item $k$: Number of 1-bits in the message, i.e., $k = \sum_{i=1}^{n} m_i$.
  \item $T$: Total number of tokens generated (i.e., the covertext length).
  \item $\delta$: Strength of the watermark perturbation. The probability of PRF-labeled “1” tokens is shifted from $p_0 = 0.5$ to
  \[
    p_w = \frac{2(1 + \delta)}{4 + \delta}, \quad \text{where } p_w > 0.5.
  \]
  \item Key usage rule: if $m_i = 1$, then key $i$ is selected with probability $1/k$ (uniformly among the 1-bits). If $m_i = 0$, then key $i$ is never used.
\end{itemize}

\subsection*{Distribution of Tokens Labeled ``1''}

For each message bit index $i$, define $s_i$ as the number of tokens in the final covertext that are labeled ``1'' by the PRF. Then:
\[
  X_i \;=\; \frac{s_i}{T}
  \;\;\; \text{is the fraction of tokens labeled 1 for bit }i.
\]

\noindent
\textbf{Case 1:} If $m_i = 0$, it is never chosen. Hence the fraction of tokens labeled ``1'' is exactly 
\[
  0.5 
  \quad\text{(no watermark shift).}
\]
That is,
\[
  \mathbb{E}[X_i] \;=\; 0.5.
\]

\noindent
\textbf{Case 2:} If $m_i = 1$, it is chosen with probability $1/k$. Whenever chosen, the fraction of tokens labeled 1 is $p_w$, otherwise it remains $0.5$. Thus:
\[
  \mathbb{E}[X_i]
  \;=\;
  \frac{1}{k} \,p_w
  \;+\;
  \Bigl(1 - \frac{1}{k}\Bigr)\,\bigl(0.5\bigr)
  \;=\;
  0.5 + \frac{1}{k}\,\bigl(p_w - 0.5\bigr).
\]

\subsection*{Binomial Z-Test for Detection}

Under the null hypothesis ($m_i = 0$), $s_i \sim \mathrm{Binomial}\bigl(T,0.5\bigr)$ approximately. Hence
\[
  \mathbb{E}[X_i] = 0.5, \quad \mathrm{Var}(X_i) = \frac{0.25}{T}.
\]
Define the Z-score for detecting a deviation:
\[
  Z_i 
  \;=\;
  \frac{\,X_i - 0.5\,}{\,\sqrt{\,0.25/T\,}}
  \;=\;
  2 \,\sqrt{T}\,\bigl(X_i - 0.5\bigr).
\]
Under the alternative hypothesis ($m_i=1$), the expected deviation is 
\[
  \Delta \;=\; \frac{1}{k}\,\bigl(p_w - 0.5\bigr).
\]
So,
\[
  \mathbb{E}[X_i] = 0.5 + \Delta
  \quad
  \text{and}
  \quad
  \mathbb{E}[Z_i]
  \;=\;
  2\,\sqrt{T}\,\Delta
  \;=\;
  2\,\sqrt{T}\,\frac{\bigl(p_w - 0.5\bigr)}{k}.
\]

\noindent Let $z_{\mathrm{th}}$ be the one-sided detection threshold. We conclude $m_i = 1$ if
\[
  Z_i \;>\; z_{\mathrm{th}},
  \quad
  \text{and otherwise conclude }m_i=0.
\]
Under the null hypothesis ($m_i=0$), $Z_i$ is approximately standard normal, so 
\[
  \mathbb{P}[\text{false positive}] \;\approx\; 1 - \Phi(z_{\mathrm{th}}).
\]
Under the alternative hypothesis ($m_i=1$), the Z-score mean shifts upward; with large $T$, the shift exceeds $z_{\mathrm{th}}$ with high probability.

\subsection*{Minimum Covertext Length $T$ for Reliable Detection}

To ensure per-bit false positive and false negative rates of at most $\alpha$, set:
\[
  z_{\mathrm{th}} \;=\;\Phi^{-1}(1-\alpha).
\]
Reliable detection requires 
\[
  \mathbb{E}[Z_i] 
  \;>\;
  z_{\mathrm{th}},
\]
which implies
\[
  T 
  \;>\;
  \frac{\;z_{\mathrm{th}}^{2}\,k^{2}\;}
       {\;4\,\bigl(p_w - 0.5\bigr)^2\;}.
\]
This ensures that each \emph{1}-bit surpasses the threshold, while each \emph{0}-bit does not with high probability.

\subsection*{Multi-Bit Recovery}

To recover all $n$ bits with overall error probability $\le \epsilon$, we can set each bit's $\alpha_{\mathrm{bit}}=\epsilon/n$. Then 
\[
  z_{\mathrm{th}}
  \;=\;
  \Phi^{-1}\!\Bigl(1 - \tfrac{\epsilon}{n}\Bigr).
\]
In the worst case ($k \le n$), set $k_{\max}=n$. Then a sufficient bound is
\[
  T
  \;>\;
  \frac{\;z_{\mathrm{th}}^2 \,\bigl(k_{\max}\bigr)^2\;}
       {\,4\,\bigl(p_w - 0.5\bigr)^2\,}
  \;=\;
  \frac{\,z_{\mathrm{th}}^{2}\,n^{2}\,}{\,4\,(p_w - 0.5)^2\,}.
\]
This ensures that for any actual $k \le n$, each \emph{1}-bit will produce a Z-score above $z_{\mathrm{th}}$, and \emph{0}-bits will not exceed $z_{\mathrm{th}}$ with high probability.

\subsection*{Example}

\begin{itemize}
\item Let $n=3$, $\delta=0.1 \implies p_w = \frac{2(1+\delta)}{4+\delta} = \frac{2.2}{4.1}\approx 0.5366.$
\item Then $p_w - 0.5\approx 0.0366.$
\item For $95\%$ success $(\epsilon=0.05)$, $\implies \alpha_{\mathrm{bit}}=0.05/3\approx0.0167$, $\implies z_{\mathrm{th}}\approx 2.17$.
\item Worst-case $k_{\max}=3$:
\[
  T
  \;>\;
  \frac{\;2.17^2 \times 3^2\;}
       {\,4\,(0.0366)^2\,}
  \;\approx\; 7900,
\]
Thus, a covertext length of approximately 8000 tokens suffices to recover all bits with high probability.
\end{itemize}
\section{Proofs for Watermarking Scheme}
\label{watermark_theorems}

We first state two key theorems for the single--bit watermarking scheme (which can be found and are proven in \cite{red_green_watermark}). We then extend these theorems to the multi--bit case.
\subsection{Single--Bit Theorems}
\label{watermark_single_bit_theorems}
\begin{theorem}[Single--Bit Correctness]
\label{sing_bit_correctness}
For the watermarking scheme presented in \cite{red_green_watermark}, if a single bit $b\in\{0,1\}$ is encoded via $st \leftarrow \text{Enc}(k,b,h)$ using key $k$, then the decoding algorithm $\text{Dec}(k,h,st)$ recovers $b$ with probability at least
\[
1-\nu_0(\lambda),
\]
where $\nu_0(\lambda)$ is a negligible function in the security parameter $\lambda$. (See \cite{red_green_watermark} for the full proof.)
\end{theorem}
\begin{theorem}[Single--Bit Steganographic Security]
\label{sing_bit_security}
For the same scheme, no PPT adversary can distinguish watermarked text from natural text with non--negligible advantage. We define
\[
\text{Adv}^{\text{cha}}_{\Pi}(\lambda) = \Bigl|\Pr\Bigl[\mathcal{A}^{\text{Enc}(k,\cdot,\cdot)}(1^\lambda)=1\Bigr] - \Pr\Bigl[\mathcal{A}^{\text{OD}(\cdot,\cdot)}(1^\lambda)=1\Bigr]\Bigr|,
\]
then
\[
\text{Adv}^{\text{cha}}_{\Pi}(\lambda) \le \nu_1(\lambda),
\]
for some negligible function $\nu_1(\lambda)$. (Again, see \cite{red_green_watermark} for details.)
\end{theorem}

\subsection{Multi-Bit Theorems}
\label{watermark_multi_bit_theorems}

In the multi--bit watermarking scheme, the message $m=(m_1,\dots, m_\ell)$ is encoded by partitioning the key as 
\[
k = k_1 \parallel k_2 \parallel \cdots \parallel k_\ell,
\]
and encoding each bit $m_i$ independently via a PRF over a recent window of tokens. We now present the proofs for correctness, steganographic security, and weak robustness.

\begin{theorem}[Correctness]
\label{watermark_multi_bit_correctness}
If an honest sender runs 
\[
st \leftarrow \text{Enc}(k, m, h)
\]
to embed the $\ell$--bit message $m=(m_1,\dots, m_\ell)$, and the receiver computes 
\[
m' \leftarrow \text{Dec}(k, h, st),
\]
then
\[
\Pr\bigl[m' \neq m\bigr] \le \ell\cdot \nu_0(\lambda),
\]
which is negligible provided $\ell$ is polynomial in $\lambda$.
\end{theorem}

\begin{proof}
By Theorem 1, each bit $m_i$ is embedded with error probability at most $\nu_0(\lambda)$. 
Since the watermarked bits are embedded independently, the probability that any bit is decoded incorrectly is at most $\nu_0(\lambda)$. By a union bound over the $\ell$ bits, the total error probability is at most 
\[
\ell \cdot \nu_0(\lambda).
\]
Hence, if the underlying single--bit scheme is correct with negligible error, then the multi--bit scheme is also correct.
\end{proof}

\begin{theorem}[Steganographic Security]
\label{watermark_multi_bit_security}
For the multi--bit watermarking scheme, no PPT adversary can distinguish watermarked text from natural text with non--negligible advantage. Formally, if
\[
\text{Adv}^{\text{cha}}_{\Pi}(\lambda) = \Bigl|\Pr\Bigl[\mathcal{A}^{\text{Enc}(k,\cdot,\cdot)}(1^\lambda)=1\Bigr] - \Pr\Bigl[\mathcal{A}^{\text{OD}(\cdot,\cdot)}(1^\lambda)=1\Bigr]\Bigr|,
\]
then $\text{Adv}^{\text{cha}}_{\Pi}(\lambda)$ is negligible in $\lambda$.
\end{theorem}

\begin{proof}
We prove the theorem via a series of hybrid experiments. Let the multi--bit message be 
\[
m=(m_1,\dots,m_\ell)
\]
and let the shared key be 
\[
k=k_1\parallel k_2\parallel \cdots \parallel k_\ell.
\]
\noindent\textbf{Hybrid \(\mathsf{H}_0\):} This is the real steganographic experiment \(\mathsf{Exp}_0\) in which the encoder generates a watermarked text using \(\mathrm{Enc}(k,m,h)\) according to the multi--bit watermarking scheme. \\

\noindent\textbf{Hybrid \(\mathsf{H}_1\):} In this experiment, we modify the challenger so that the watermark embedding for the \emph{first} bit (associated with key \(k_1\) and message bit \(m_1\)) is replaced by sampling tokens from the natural distribution (i.e., without applying the watermark bias). The embedding procedures for \(m_2,\dots,m_\ell\) remain unchanged. By the single--bit steganographic security (Theorem 2), the distinguishing advantage between \(\mathsf{H}_0\) and \(\mathsf{H}_1\) is at most \(\nu_1(\lambda)\). \\

\noindent\textbf{Hybrid \(\mathsf{H}_2\):} In \(\mathsf{H}_2\), we modify the challenger further so that the watermark embedding for the \emph{first two} bits (associated with keys \(k_1\) and \(k_2\)) is replaced by natural token sampling, while the embedding for \(m_3,\dots,m_\ell\) remains unchanged. Again, by the security of the single--bit scheme, the advantage between \(\mathsf{H}_1\) and \(\mathsf{H}_2\) is at most \(\nu_1(\lambda)\). \\

\noindent\textbf{Hybrid \(\mathsf{H}_i\) for \(1\le i\le \ell\):} More generally, define \(\mathsf{H}_i\) to be the experiment in which the watermark embedding for the first \(i\) bits is replaced by natural token sampling, while the remaining \(\ell-i\) bits are embedded normally. By the single--bit security guarantee, for every \(i\) the distinguishing advantage between \(\mathsf{H}_{i-1}\) and \(\mathsf{H}_i\) is at most \(\nu_1(\lambda)\). \\

\noindent\textbf{Hybrid \(\mathsf{H}_\ell\):} In the final hybrid, all watermarking instances are replaced by natural token sampling. Hence, \(\mathsf{H}_\ell\) is identical to the natural text generation process, with no watermark present. \\

By the triangle inequality, the overall distinguishing advantage between \(\mathsf{H}_0\) (the real watermarked experiment) and \(\mathsf{H}_\ell\) (the natural text generation) is at most
\[
\sum_{i=1}^{\ell} \nu_1(\lambda)=\ell\cdot \nu_1(\lambda).
\]

Since \(\ell\) is polynomial in \(\lambda\) and \(\nu_1(\lambda)\) is negligible, the overall advantage is negligible. \\

Thus, no PPT adversary can distinguish watermarked text from natural text with more than negligible advantage, proving the steganographic security of the multi--bit scheme.
\end{proof}

\begin{theorem}[Weak Robustness]
\label{watermark_multi_bit_weak_robustness}
Under the weak threat model (see Section~\ref{weak_threat_model}), suppose an adversary applies a tampering function
\[
f \in \mathcal{F}_{k,\epsilon}.
\]
Then, for an honestly generated stegotext
\[
st \leftarrow \text{Enc}(k, m, h),
\]
we have
\[
\Pr\Bigl[\text{Dec}(k, h, f(st)) \neq m\Bigr] \le \ell\cdot \nu_0(\lambda) + \ell\cdot \eta(\epsilon),
\]
where $\eta(\epsilon)$ is a function that vanishes as $\epsilon \to 1$. In other words, $ \nu_0(\lambda)$ captures the inherent error in recovering watermarked text and $\eta(\epsilon)$ captures the additional error introduced by adversarial modifications that disrupt the watermark’s structure. As the tampering becomes increasingly mild (i.e., as $\epsilon$ approaches 1 and nearly all $k$-length substrings are preserved), this additional error becomes negligible. In particular, if the tampering is sufficiently mild, the message $m$ is recovered with overwhelming probability.
\end{theorem}

\begin{proof}
In \cite{red_green_watermark}, it is shown that when a single watermark bit is embedded via the watermarking scheme, if an adversary applies a tampering function that is \((k,\epsilon)\)–locally consistent (i.e., preserving at least an \(\epsilon\) fraction of the \(k\)–length substrings), then the detection statistic (such as the z-score) remains above the detection threshold with high probability. In particular, there exists a negligible function \(\nu_0(\lambda)\) (arising from the inherent randomness of the watermarking process) and an additional term \(\eta_i(\epsilon)\) (which quantifies the extra error due solely to the adversary's modifications) such that the probability that the single watermark bit \(m_i\) is misdecoded is at most
\[
\Pr[\text{Dec}_i(k_i,h,f(st)) \neq m_i] \le \nu_0(\lambda) + \eta_i(\epsilon).
\]
Here, \(\eta_i(\epsilon)\) depends on the specific nature of the tampering; importantly, \(\eta_i(\epsilon) \to 0\) as \(\epsilon \to 1\) because when nearly all \(k\)-substrings are preserved, the adversarial modifications have a negligible effect on the watermark’s structure.
Assuming that the same robustness guarantees hold uniformly for each watermark bit (i.e., \(\eta_i(\epsilon) \le \eta(\epsilon)\) for all \(i\)), we can apply a union bound over the \(\ell\) bits of the hidden message \(m=(m_1,\dots,m_\ell)\). Thus, the overall probability that at least one bit is misdecoded is bounded by
\[
\Pr\Bigl[\text{Dec}(k, h, f(st)) \neq m\Bigr] \le \ell\cdot \nu_0(\lambda) + \ell\cdot \eta(\epsilon).
\]
This shows that under tampering by any \(f \in \mathcal{F}_{k,\epsilon}\), as long as \(\epsilon\) is close to 1 (i.e., the adversary’s modifications are sufficiently mild), the extra error \(\eta(\epsilon)\) becomes negligible. Hence, the overall scheme remains robust---the hidden message is recovered with overwhelming probability.
\end{proof}

\section{Proofs for Embedding Scheme}
\label{embedding_theorems}

\begin{theorem}[Correctness]
\label{embedding_correctness_proof}
Let $m\in\{0,1\}^n$ be the hidden message partitioned into $r=\lceil n/l\rceil$ chunks, $E$ be an embedding function, and $H$ be a locally sensitive hash function that outputs $l$ bits where \(H(E(x)) = H(E(x'))\) if \(\text{distance}(E(x), E(x')) \leq \delta\) for $x, x' \in \{0,1\}^n$ and \(\Pr \left[H(y) = 1\right] = 1/2\) for a randomly sampled $y\in\{0,1\}^n$. Additionally, assume the language model being used produces enough entropy in each sample that the encoder can always sample a valid token. If for each chunk $m_i$, the encoder outputs a candidate text $x$ such that 
\[
H(E(x)) = m_i,
\]
then the decoding algorithm recovers the entire message $m$ with overwhelming probability (i.e., the error probability is negligible in the security parameter $\lambda$).
\end{theorem}

\begin{proof}
For a fixed chunk $m_i$, since the LSH function is balanced, a randomly generated candidate text $x$ from the language model satisfies
\[
\Pr\bigl[H(E(x)) = m_i\bigr] = \frac{1}{2^l}.
\]
Thus, the expected number of samples needed is $2^l$. Given the high--entropy assumption of the language model, the probability that the encoder fails to find a valid candidate within a polynomial number of samples is negligible. By applying a union bound over all $r$ chunks, the overall failure probability remains negligible when $r$ is polynomial in $\lambda$. Therefore, an honestly encoded message is decoded correctly with overwhelming probability.
\end{proof}

\begin{theorem}[Undetectability]
Assume that the LSH function $H$ evenly splits the message space (i.e., each $l$--bit outcome occurs with probability $1/2^l$) and that the PRF used in the perturbation is keyed so that its output is unpredictable. Then, provided that the language model produces high--entropy candidate texts, the rejection sampling procedure yields stegotexts that are statistically indistinguishable from naturally generated texts.\footnote{We ignore low-entropy cases in this proof. In practice, this is handled by sampling incorrect candidates and relying on error correction during recovery.}
\end{theorem}

\begin{proof}
The encoder repeatedly samples candidate texts $x$ until $H(E(x)) = m_i$ for the desired chunk $m_i$. Since $H$ is balanced, $H(E(x)) = m_i$ happens with probability $1/2^l$, and conditioning on this event does not significantly alter the underlying distribution of $x$. Additionally, the perturbations induced by the keyed PRF (which determines the split) are unpredictable to PPT adversaries. Therefore, the distribution of sampled texts remains statistically close to the natural distribution of texts generated by the language model. Therefore, no efficient adversary can distinguish between stegotexts and naturally generated texts with non--negligible advantage.
\end{proof}

\begin{theorem}[Robustness]
\label{thm:robustness}
Let \(m = (m_{1}, \dots, m_{r})\) be a hidden message, and let \(x = (x_{1}, \dots, x_{r})\) be the covertext in which the \(r\) chunks of \(m\) are respectively embedded. Suppose an adversary obtains \(x\) and produces
\[
x' = (x'_1, \dots, x'_r)
\]
by paraphrasing some subset of the covertext chunks, subject to the condition
\[
\mathrm{distance}(E(x_j), E(x'_j)) \le \delta
\quad
\text{for each paraphrased chunk } x'_j
\]
where $j \in \{1 \dots n\}$.
Then, by the design of the LSH function \(H\), the decoding algorithm recovers all \(r\) chunks of \(m\) with overwhelming probability.
\end{theorem}

\begin{proof}
By strong robustness (Definition~\ref{def:StrongRobustness}), whenever a chunk \(x_j\) is paraphrased into \(x'_j\) with
\[
\mathrm{distance}(E(x_j), E(x'_j)) \le \delta,
\]
the LSH function \(H\) is designed so that
\[
H(E(x'_j)) = H(E(x_j))
\]
with probability at least \(1 - \nu(\lambda)\) for some negligible function \(\nu\). Consequently, the decoded chunk \(m_j\) remains correct with probability at least \(1 - \nu(\lambda)\) for each $j \in \{1 \dots n\}$.

\medskip

\noindent
Define the event \(F_j\) to be “chunk \(m_j\) fails to decode after paraphrasing.” Then
\[
\Pr[F_j] \le \nu(\lambda).
\]
Let \(F = (F_1 \cup F_2 \cup \dots \cup F_r)\) be the event that at least one chunk fails. By the union bound,
\[
\Pr[F]
= \Pr(F_1 \cup \cdots \cup F_r)
\le \sum_{j=1}^r \Pr[F_j]
\le r \cdot \nu(\lambda).
\]
Hence the probability of successfully decoding all \(r\) chunks simultaneously is
\[
\Pr(\text{all chunks decode correctly})
\ge 1 - r \cdot \nu(\lambda).
\]
Since \(r\) is polynomial in the security parameter \(\lambda\), and \(\nu(\lambda)\) is negligible, the product \(r \cdot \nu(\lambda)\) is also negligible. Therefore, with overwhelming probability, all chunks of \(m\) are recovered correctly from the paraphrased covertext \(x'\). This establishes the scheme’s robustness against paraphrasing of any subset of chunks.
\end{proof}

\section{PCA Prompt}
\label{pca_prompt}

We trained a variant of our PCA LSH over a cleaned variant of the Enron corporate email dataset called the Enronsent corpus. \cite{styler2011enronsent}. The following prompt is designed to be used with this model and can be found at \url{https://github.com/NeilAPerry/robust_steganography}. \\

\noindent `````` \\
\noindent I wanted to follow up regarding the implementation timeline for the new risk management system. Based on our initial assessment, we'll need to coordinate closely with both IT and Operations to ensure a smooth transition. Please review the attached documentation when you have a moment. \\
    
\noindent After consulting with the development team, we've identified several key milestones that need to be addressed before proceeding. The current testing phase has revealed some potential integration issues with our legacy systems, particularly in the trade validation module. We're working on implementing the necessary fixes and expect to have an updated timeline by end of week. \\
    
\noindent Given the complexity of these changes, I believe it would be beneficial to schedule a stakeholder review meeting. We should include representatives from Risk Management, IT Operations, and the Trading desk to ensure all requirements are being met. I've asked Sarah to coordinate calendars for next Tuesday afternoon. \\
\noindent '''''' \\

\end{document}